\documentclass[a4paper,11pt]{article}
\pdfoutput=1 

\usepackage[utf8]{inputenc}
\usepackage{jheppub}
\usepackage{natbib}
\usepackage{ulem}

\usepackage{float}

\usepackage{amssymb,amsmath,amsbsy,amsthm,mathtools}
\usepackage{braket}
\usepackage[mathscr]{euscript}

\newtheorem{definition}{Definition}[section]

\newtheorem{lemma}{Lemma}[section]
\newtheorem{theorem}{Theorem}[section]
\newtheorem{corollary}{Corollary}[section]

\theoremstyle{definition}
\newtheorem{remark}{Remark}

\usepackage{todonotes}
\usepackage{graphicx,color}
\usepackage{multirow}
\usepackage{hyperref}
\usepackage{subcaption}
\newcommand{\mC}{\mathcal{C}}

\newcommand{\mE}{\mathcal{E}}

\newcommand{\mP}{\mathcal{P}}

\newcommand{\msR}{\mathscr{R}}

\newcommand{\lp}{\left(}
\newcommand{\rp}{\right)}

\title{Combinatorial aspects of holographic quantum secret sharing}

\author[a,b]{Ning Bao}
\author[a]{Keiichiro Furuya}
\author[a]{Jacob March}

\affiliation[a]{Department of Physics, Northeastern University, Boston, MA, 02115, USA}
\affiliation[b]{Computational Science Initiative, Brookhaven National Laboratory, Upton, NY 11973 USA}

\emailAdd{ningbao75@gmail.com}
\emailAdd{k.furuya@northeastern.edu}
\emailAdd{j.march@northeastern.edu}

\makeatletter
\gdef\@fpheader{}
\makeatother

\abstract{We introduce combinatorial holographic quantum secret sharing (CHQSS) for a bulk subregion in AdS$_3$/CFT$_2$ to study how logical information of the bulk subregion is encoded in the boundary and protected from erasures of boundary subregions. We introduce a distance, a reconstruction threshold, and a secret threshold to characterize CHQSS schemes. We present the phase transitions of multipartite entanglement wedges in a symmetric setup and observe multiple distinct phase transition points. The distance and thresholds depend on the holographic phase and the choice of bulk subregion. We derive the maximum distance. Moreover, we derive the relations between the distance and the thresholds. We construct a family of CHQSS schemes in the symmetric setting that includes perfect threshold CHQSS and perfect non-threshold CHQSS.}

\begin{document}
\maketitle


    

\section{Introduction}


In quantum information processing, encoding logical information is essential for protecting it from errors and achieving quantum tasks. For example, one prepares a quantum error-correction code with a choice of encoding so that the encoded logical information is undisturbed by a set of errors and recoverable.

Secret sharing is an encoding scheme that makes secret information accessible only to certain sets of shares. Those shares that can jointly reconstruct the secret are called authorized sets, and the ones that cannot are called unauthorized sets. An access structure is a set of authorized sets, and a forbidden structure is a set of unauthorized sets. Various types of classical and quantum secret sharing schemes \cite{Gottesman:1999jzr,Cleve:1999qg,Hillery:1998yq,Smith:2000wvp,Zhang:2014jqy,Xiao:2004gua,Cakan:2023izz,cryptoeprint:2025/518,Qin:2007eih,Gheorghiu:2012obj,Applebaum2026secretsharing,Lie2018secretsharing,cryptoeprint:2012/412,Helwig:2012nha,Markham:2008kqi,Chiwaki:2025erf,Matsumoto:2019wct,Matsumoto:2017xuk,Javelle:2011dpw,Sarvepalli:2012nfk,Keet:2010rlc} have been proposed and studied. It has been proven that some of them are equivalent to erasure codes \cite{Cleve:1999qg}. In short, secret sharing schemes can be characterized by examining how easy or difficult it is to reconstruct the logical information.

In a holographic setup, a quantum secret sharing scheme is reminiscent of bulk reconstruction or holographic encoding, which maps bulk degrees of freedom into a boundary theory. In other words, for instance, in the entanglement wedge reconstruction \cite{Jafferis:2015del, Dong:2016eik,Harlow:2016vwg,Pastawski:2015qua,Almheiri:2014lwa,Furuya:2020tzv,Leutheusser:2022bgi,Cotler:2017erl,Chen:2019gbt}, a boundary subregion is an authorized set that can reconstruct the secret in a bulk subregion if the corresponding entanglement wedge of the boundary subregion contains the bulk subregion. The complement of the boundary subregion with respect to the entire boundary is an unauthorized set.

As studied in \cite{Gottesman:1999jzr}, a quantum secret sharing scheme exists for any access structure if and only if it satisfies a monotonicity structure and the no-cloning theorem. We will see in the main text that the access structure can be defined in the holographic setup by the entanglement wedge reconstruction, and that the two properties, the monotonicity and the no-cloning theorem, follow from entanglement wedge nesting \cite{Dong:2016eik,Akers:2017ttv,Czech2025EW,Czech:2023xed,Akers:2023fqr,Bao2025entanglementwedge} together with the complementarity of entanglement wedges. Thus, the quantum secret sharing scheme naturally arises in the holographic setup, as mentioned in the papers \cite{Almheiri:2014lwa,Pastawski:2015qua}. 

We should note that quantum error correction and quantum secret sharing capture different aspects of encoding, although encoding is a key process in both. Holographic quantum error correction (HQEC) studies how encoded bulk logical information is protected against erasures of boundary subregions, whereas holographic quantum secret-sharing (HQSS) focuses on how bulk logical information is encoded to the holographic states on the boundary, that is, how easy or difficult it is to reconstruct, and how secret it is. 

Holographic quantum secret sharing can be either a pure-state or mixed-state scheme based on how many boundary regions one has access to. 

On the one hand, we have a pure-state scheme if we have access to the entire boundary for the bulk reconstruction. We then have infinitely many choices of boundary subregions to reconstruct the bulk logical information in $b$.
Moreover, the logical information in $b$ is protected against any boundary erasure unless the entanglement wedge of the erased boundary subregion contains $b$. 

On the other hand, we have a mixed-state QSS if we restrict ourselves to a part of the boundary. In this case, achieving the bulk reconstruction of the bulk logical information in $b$ depends on the shape and the configuration of the corresponding entanglement wedges. 

Especially for holographic mixed states, holographic phase transitions and redundancy in the holographic encoding affect the quality of holographic quantum secret sharing associated with a bulk subregion.

\begin{figure}
    \centering
    \includegraphics[width=0.9\linewidth]{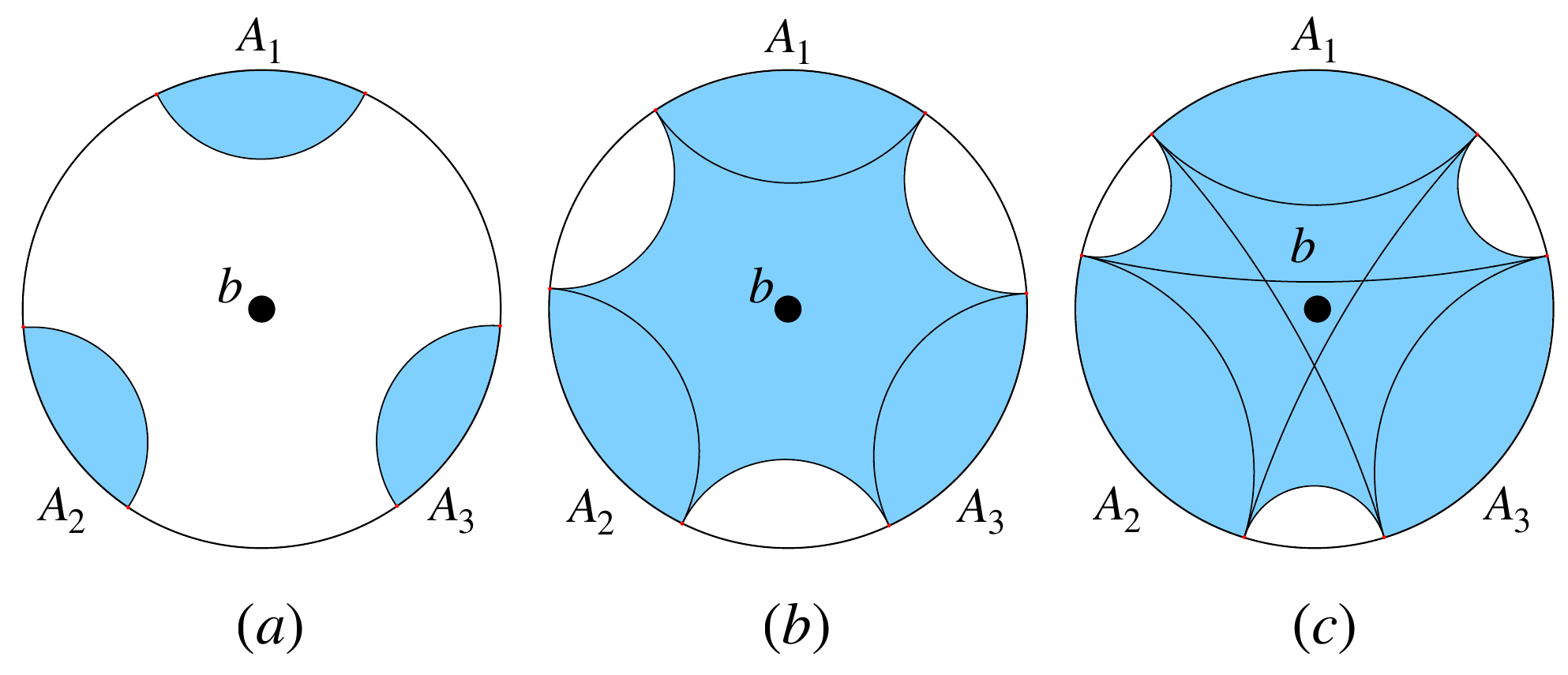}
    \caption{\small{Three distinct holographic phases for $n=3$ boundary subregions. The RT surfaces of the element in $\mP[n]=\{A_1,A_2,A_3,A_1A_2, A_1A_3, A_2A_3, A_1A_2A_3\}$ are drawn with a black line. The entanglement wedge of $A_1A_2A_3$ is highlighted blue. The black dot represents a spatial point in the bulk subregion $b$: (a) Disconnected phase, (b) Only the tripartite entanglement wedge, $EW(A_1A_2A_3)$, is connected, (c) The entanglement wedges of all combinations of boundary subregions are connected.}}
    \label{fig:geometry_n-3}
\end{figure}

In figure \ref{fig:geometry_n-3}, we have three distinct holographic phases for $n=3$ boundary subregions, $[n=3]:=\{A_1,A_2,A_3\}$. Reconstructability of the bulk logical information at a representative spatial point $b$ in the bulk subregion, represented as a black dot, from the boundary subregions $A_1,A_2,A_3$ in the figure depends on a holographic phase. 

In the case of $(a)$, the bulk logical information at $b$ cannot be reconstructed from any combination of $[n]$. However, it is reconstructible in the case of figure \ref{fig:geometry_n-3}-(b) and (c). In (b), simultaneous access to three boundary subregions, such as $A_1A_2A_3$, is required to jointly reconstruct the bulk logical information. In (c), at least two boundary subregions are required to reconstruct the bulk logical information. In both (b) and (c), none of the single boundary subregions can reconstruct the bulk logical information.


The redundancy of the holographic encoding matters when considering erasures of boundary subregions. We say that bulk logical information is \textit{redundantly} encoded\footnote{Its definition is given in definition \ref{def:redundant_encoding}.} against the erasure if the logical information remains reconstructible after the erasure. For example, in the case of $(c)$, the bulk subregion $b$ is redundantly encoded because the bulk logical information in $b$ can still be reconstructed from $A_1A_2$ even if one loses access to $A_3$. However, the bulk logical information at $b$ in figure \ref{fig:geometry_n-3}-(b) is encoded less redundantly and can never be retrieved if the boundary subregion $A_3$ is erased in the holographic phase, where only the tripartite entanglement wedge is connected. 



Therefore, there are apparent differences in the quality of holographic encoding between pure and mixed states, arising from holographic phase transitions and the connectivity of entanglement wedges. Particularly in the case of mixed-state holographic encoding, the existence of bulk logical information that can be reconstructed from $EW(R_1\cup R_2)$, but not from each individual $EW(R_1)$ and $EW(R_2)$, has been observed as an obstruction for the holographic QEC to be exact \cite{Faulkner:2020hzi,Leutheusser2025quantumtasksQEC}.  

The main goal of this paper is to characterize the quality of holographic encoding in terms of quantum secret sharing and to understand the difference between pure-state and mixed-state holographic encoding, including the obstruction to the exactness of the holographic QEC.

To study the holographic quantum secret sharing, we consider the following coarse-grained situation. For a set $[n]:=\{A_1,\cdots, A_n\}$ of $n$ boundary subregions, consider a certain phase of entanglement wedges of $\mP[n]$, the power set of $[n]$. Then, the dual bulk geometry is partitioned by the RT surfaces. We coarse-grain the holographic geometry by introducing the RT-region graph studied in \cite{Bao:2015bfa,Bao:2024azn,Bao:2025sjn}. Each bulk subregion enclosed by portions of RT surfaces is represented as a vertex. Two vertices are connected by an edge if the geodesics between the vertices cross an RT surface once. 

For each vertex $b$ in the region graph, we construct a set $\msR_b$ of boundary subregions whose entanglement wedge contains $b$. The set $\msR_b$ is the access structure that satisfies the monotonicity and no-cloning theorem. With the access structure, we define \textit{combinatorial holographic quantum secret sharing} (CHQSS). We aim to understand how easy or difficult it is to reconstruct the bulk logical information and how well it is protected against erasures of boundary subregions.

As in $(\!(r,s,n)\!)$-QSS\footnote{A perfect threshold QSS is often written as $(\!(r,n)\!)$-QSS without $s$ because $r=s+1$\cite{Cleve:1999qg,Gottesman:1999jzr}.} for a finite quantum system, we define a reconstruction threshold $r_b$ and a secret threshold $s_b$ for the bulk subregion $b$ in our setup and characterize the scheme by them, i.e., $(\!(r_b,s_b,n)\!)$-CHQSS. Intuitively, the reconstruction threshold measures how easy it is to reconstruct the bulk logical information, and the secret threshold measures how difficult it is to reconstruct. In addition, we define the distance $d_b$, which measures how well the bulk logical information in $b$ is protected against the erasure errors. The distance is analogous to the distance defined in the study of quantum error-correction codes. We derive the relations among $r_b$, $s_b$, and $d_b$ in theorem \ref{thm:reconstruction_distance} and \ref{thm:secret_distance}.




In section \ref{sec:transition_points_MEW}, we first present that, for a symmetric case, the holographic phase transition of multipartite entanglement wedges can occur at different transition points. In section \ref{sec:redundant_encoding}, we give a precise definition of the redundancy of the holographic encoding. We also introduce the distance $d_b$ and find the maximum distance as a function of $n$. 
In section \ref{sec:CHQSS}, we construct combinatorial holographic quantum secret sharing. After defining the reconstruction and secret thresholds, we evaluate them for a choice of $n$ and a bulk subregion $b$ in a given holographic phase. We then explore CHQSS in terms of several types of known QSS, such as threshold/non-threshold QSS. We conclude the paper with discussions in section \ref{sec:discussions}.

\section{Holographic phase transitions of multipartite entanglement wedges}\label{sec:transition_points_MEW}

In this section, we present examples of holographic phase-transition points of multipartite entanglement wedges. We consider a symmetric setup where there are $n$ boundary subregions with the same angular interval $\phi$. Moreover, they are evenly spaced along the boundary. We observe several distinct transition points for multipartite entanglement wedges.

\subsection{Symmetric cases}


In pure AdS$_3$, the area of the RT surface of a boundary subregion $A$ whose angular interval $\phi$ is
\begin{equation}
    \gamma_A = 2\log \frac{2}{\varepsilon}\sin\lp\frac{\phi}{2}\rp,
\end{equation}
in the units of the AdS radius, where $\varepsilon$ is a UV regulator, $\phi$ is an angular interval, measured in radians, of the subregion $A$, see, for instance, figure \ref{fig:angular_interval-n-5}.

We consider a set $[n]$ of $n$ boundary subregions in pure AdS$_3$. The angular size of each subregion is denoted by $\phi$, and all of them are evenly distributed on the boundary. If the union of $n$ boundary subregions does not cover the entire boundary, we have the complement of the union as a purifier $O$. 
We label the boundary subregions in counterclockwise order, with subscripts increasing along the boundary, see figure \ref {fig:angular_interval-n-5}. 

\begin{figure}
    \centering
    \includegraphics[width=0.5\linewidth]{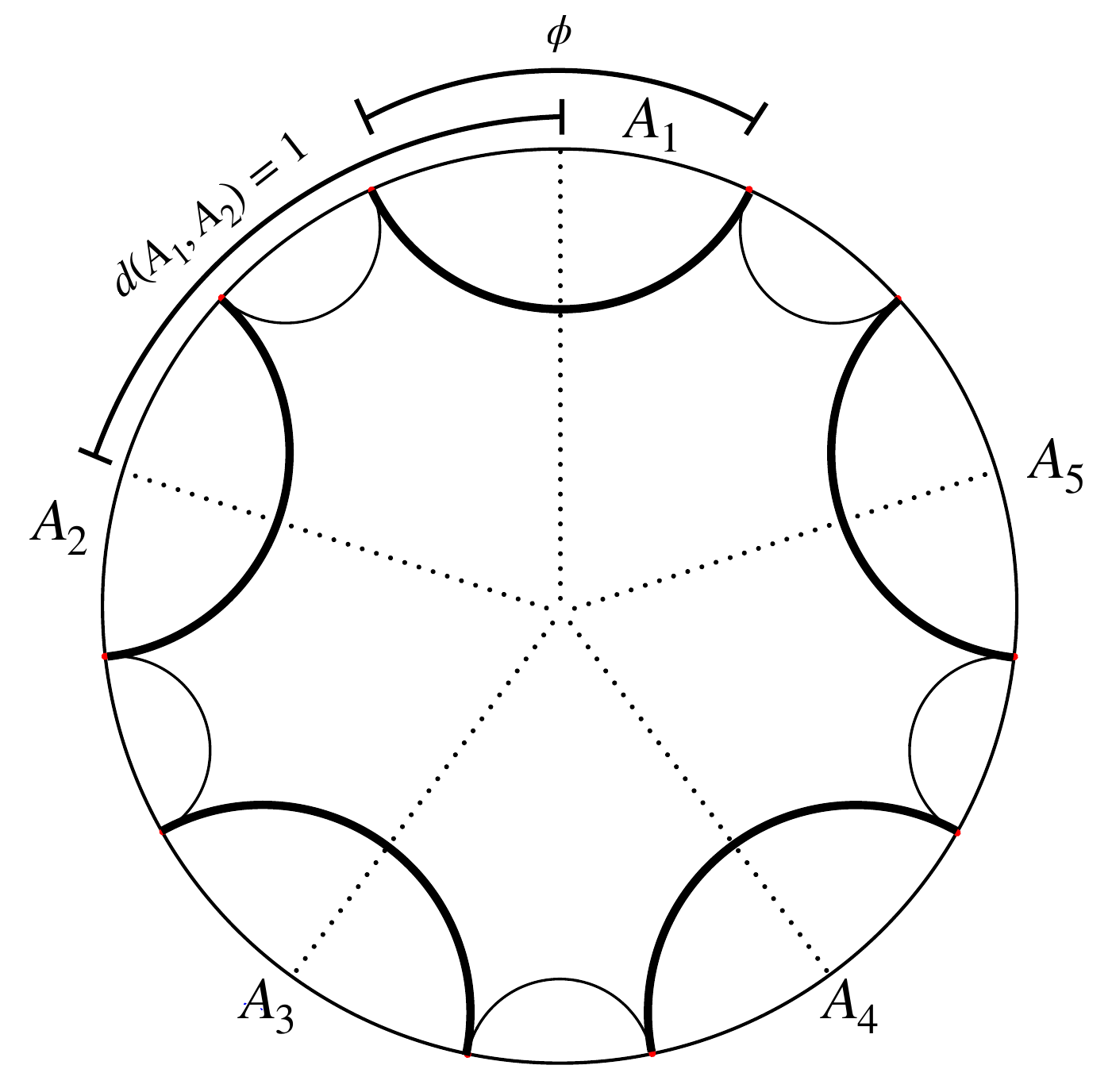}
    \caption{\small{The symmetric setup for $n=5$ on a constant time slice of AdS$_3$. The entanglement wedges of $A_1$, $A_2$, $A_3$, $A_4$, $A_5$, and $A_1A_2A_3A_4A_5$ are drawn. The angular interval of each boundary subregion is $\phi$. The directed distance between $A_1$ and $A_2$ is $d(A_1,A_2)=1$. }}
    \label{fig:angular_interval-n-5}
\end{figure}

When computing an entanglement entropy of $m$-partite boundary subregions, we need to choose the configuration of entanglement wedges, which gives the minimum entanglement entropy. We denote a set of the boundary subregions of an $m$-partite entanglement wedge by $\{A_{i_p}\}_{p=1}^{m} \subseteq [n]$ where $i_p \in \{1,\cdots, n\}$ and we adopt the cyclicity $A_{i_{n+1}}:=A_{i_1}$. For example, $\{A_{i_1} = A_1,A_{i_2}=A_3,A_{i_3}=A_4\}\subset [n=5]$ is a set of boundary subregions for a tripartite entanglement wedge.

We compute the entanglement entropy of $\{A_{i_p}\}_{p=1}^m$ 
\begin{equation}
    S_{i_1 \cdots i_m} = \frac{\gamma_{i_1 \cdots i_m}}{4G_N},
\end{equation}
where $\gamma_{i_1 \cdots i_m}$ is the total area of the minimal configuration of RT surfaces homologous to the union of the $m$ boundary subregions.


The phase transition occurs when the configuration of the entanglement wedges switches from one to another\footnote{See \cite{Ju:2024hba,Ju:2025tgg} for the studies of holographic multipartite entanglement in multipartite entanglement wedges and their connectivity criteria.}. We have numerically computed the transition points for $n=3,4,5$, see figure \ref{fig:transition_points_n-3-4} and \ref{fig:transition_geometry_n-5-221}.

\begin{figure}[t]
    \centering
    \includegraphics[width=1.0\linewidth]{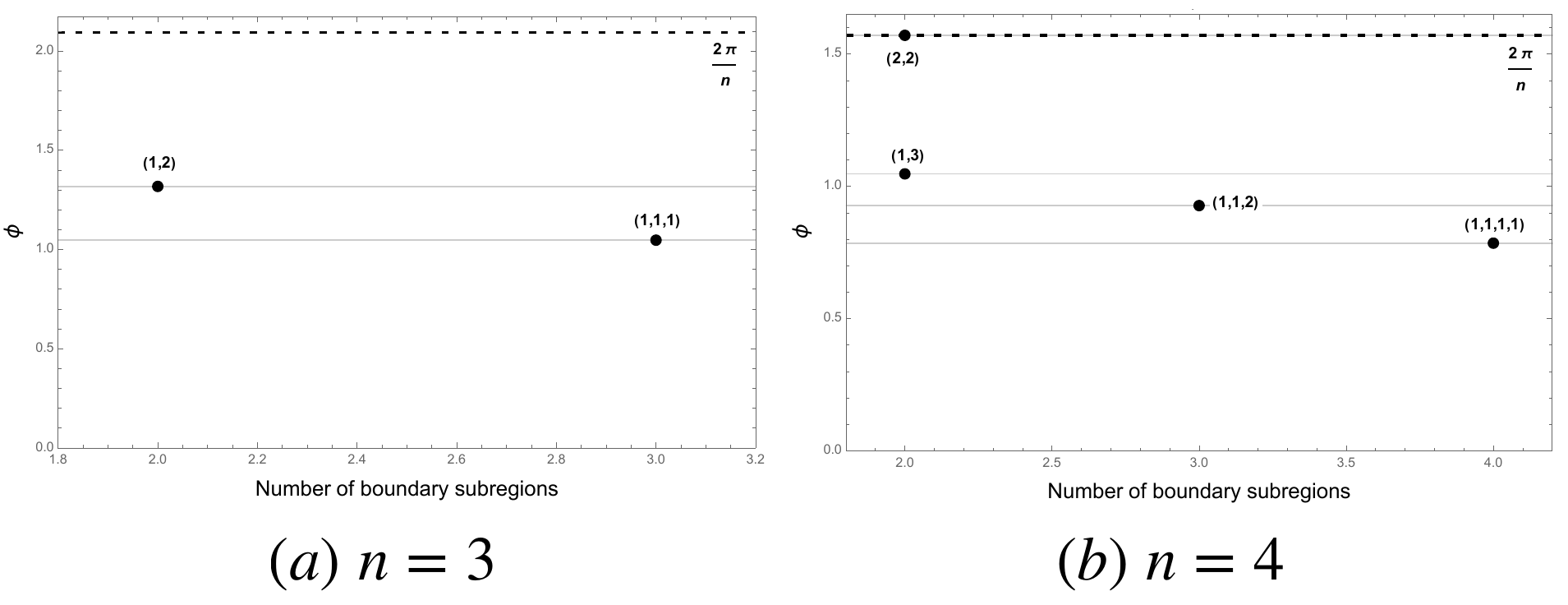}
    \caption{\small{The plots of phase transition points for $n=3,4$. The horizontal axis $m$ is the number of boundary subregions for $m$-partite entanglement wedges. The vertical axis is the angular interval $0\leq \phi \leq 2\pi/n$. We put the black dotted line when $\phi = 2\pi/n$. The gray lines represent the value of transition points $\phi_t$. For $n=3$, we have $\phi_0=1.05\sim \pi/3$ and $\phi_1=1.32$ radians. For $n=4$, we have $\phi_0=0.785$, $\phi_1=0.927$, $\phi_2=1.05\sim \pi/3$ and $\phi_3=1.57 \sim 2\pi/4$ radians. The lowest gray horizontal line corresponds to the transition point $\phi_{t=0}$ from the fully disconnected phase to the first connected phase. The labels $(x_{i_1},\cdots, x_{i_3})$ for $i_p\in \{1,2,3\}$ and $(x_{i_1},\cdots, x_{i_4})$ for $i_p\in \{1,2,3,4\}$ are the sets of boundary subregions whose entanglement wedges are in their connected phase after $\phi_t$.}}
    \label{fig:transition_points_n-3-4}
\end{figure}

We introduce the label $(x_{i_1},\cdots, x_{i_{m}})$ for a set of boundary subregions whose $m$-partite entanglement wedge is connected. Here, $x_{i_m}$ denotes the spacing pattern defined below. We define a directed distance function $d$ in the counterclockwise direction in the unit of $2\pi/n$ between the boundary subregions $A_i,A_j$ as, see figure \ref{fig:angular_interval-n-5},
\begin{equation}\label{def:set_distance}
    d (A_i, A_j) := (j-i) \mod n.
\end{equation}
For example, for $n=5$, 
\begin{equation}
    d (A_1, A_3) = 2 \mod 5 = 2,
\end{equation}
whereas
\begin{equation}
    d (A_3, A_1) = -2 \mod 5 = 3.
\end{equation}

With the distance function, we write 
\begin{equation}
    x_{i_p} := d(A_{i_{p}},A_{i_{p+1}})    
\end{equation}
for $p=1,\cdots,m$. 
Moreover, $x_{i_p}$ satisfies 
\begin{equation}
    \sum_{p} x_{i_p} = n.
\end{equation}
Note that the set of boundary subregions represented by $(x_{i_1},\cdots, x_{i_{m}})$ depends only on the set $\{x_{i_1},\cdots, x_{i_{m}}\}$.  

For example, consider $n=5$ and $m=3$-partite entanglement wedges. There are two sets of boundary subregions represented by $(1,2,2)$ and $(1,1,3)$, i.e.,
\begin{equation}
    (1,2,2) = \{A_1A_2A_4, A_2A_3A_5,A_3A_4A_1,A_4A_5A_2,A_5A_1A_3\},
\end{equation}
and 
\begin{equation}
    (1,1,3) = \{A_1A_2A_3, A_2A_3A_4,A_3A_4A_5,A_4A_5A_1,A_5A_1A_2\}.
\end{equation}
In particular, $A_3A_4A_1$ is the element of the set of boundary subregions represented by $(1,2,2)$ because
\begin{equation}
    d(A_3,A_4) = 1,\; d(A_4,A_1) = 2,\;d(A_1,A_3) = 2,
\end{equation}
see figure \ref{fig:transition_geometry_n-5-221}-(b).

\subsubsection{Examples}

\begin{figure}
    \centering
    \includegraphics[width=1.0\linewidth]{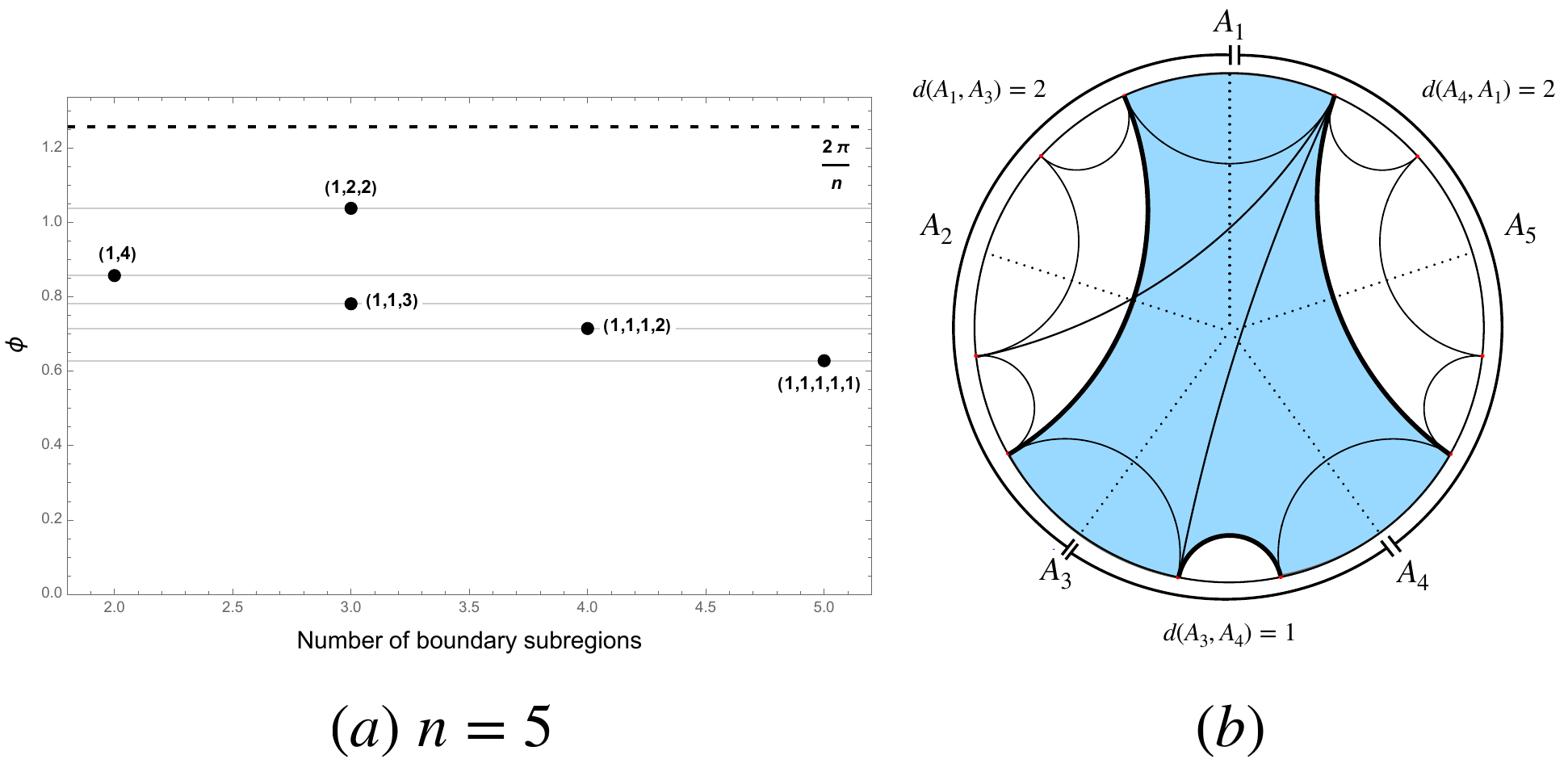}
    \caption{\small{(a) The plot for the transition points for $n=5$. We have five transition points, $\phi_0=0.628\sim \pi/5$, $\phi_1=0.715$, $\phi_2=0.781$, $\phi_3=0.857$, $\phi_4=1.04$ radians. (b) The configuration of entanglement wedges of $A_1$, $A_2$, $A_3$, $A_4$, $A_5$, $A_1A_2$, $A_1A_2A_3$, $A_1A_2A_3A_4$, $A_1A_2A_3A_4A_5$, and $A_1A_3A_4$ when $\phi \geq \phi_{t=4}$ on a constant time slice of AdS$_3$. The entanglement wedge of $A_1A_3A_4$ is highlighted blue.}}
    \label{fig:transition_geometry_n-5-221}
\end{figure}

Let us present the example cases for $n=3, 4, 5$.  
From now on, we label the distinct holographic phase transition points by $t=0, 1,\cdots$. The order of the integers, e.g., $t=0 < t=1$, is consistent with the ordering of the angular intervals, e.g., $\phi_{t=0}<\phi_{t=1}$. We set $t=0$ as the transition point from a fully disconnected phase to the first connected phase. 

For $n=3$, there are two transition points and three phases, see figure \ref{fig:geometry_n-3} and \ref{fig:transition_points_n-3-4}. The first transition is from the fully disconnected phase into the connected phase of the tripartite entanglement wedge whose boundary subregion is represented by $(1,1,1) = \{A_1A_2A_3\}$. It should be noted that the bipartite connected entanglement wedges are not connected in the phase $\phi_{t=0} \leq \phi < \phi_{t=1} $. When $\phi\geq \phi_{t=1}$, all of the bipartite entanglement wedges are connected in addition to the tripartite entanglement wedge. Those bipartite entanglement wedges are represented as $(1,2) = \{A_1A_2, A_2A_3,A_3A_1\}$.

The case $n=4$ has four transition points $t=0,\cdots,3$ in figure \ref{fig:transition_points_n-3-4}. We see that two types of bipartite entanglement wedges enter the connected phase at different transition points. The set of boundary subregions of the first type is represented by $(1,3) = \{A_1A_2,A_2A_3,A_3A_4,A_4A_1\}$. The second type is represented by $(2,2) = \{A_1A_3,A_2A_4\}$. The entanglement wedges of the boundary subregions represented by $(2,2)$ are disconnected until the last transition point $\phi_{t=3}$.

The case $n=5$ in figure \ref{fig:transition_geometry_n-5-221} has multiple distinct shapes of bipartite and tripartite entanglement wedges. The bipartite entanglement wedges whose boundary subregions represented by $(1,4)=\{A_1A_2, A_2A_3, A_3A_4,A_4A_5,A_5A_1\}$, get connected. However, the bipartite entanglement wedges of the boundary subregions represented by 
\begin{equation}
    (2,3) = \{A_1A_3,A_1A_4,A_2A_4,A_2A_5,A_3A_5\} 
\end{equation}
do not get connected before $\phi=2\pi/5$, which is the upper bound of the angular interval in the current symmetric setup. 

All possible tripartite entanglement wedges get connected at the different transition points. The sets of corresponding boundary subregions are
\begin{equation}
    (1,2,2) = \{A_1A_2AA_4, A_2A_3A_5,A_3A_4A_1,A_4A_5A_2,A_5A_1A_3\},    
\end{equation}
and
\begin{equation}
    (1,1,3) = \{A_1A_2A_3, A_2A_3A_4,A_3A_4A_5,A_4A_5A_1,A_5A_1A_2\}.   
\end{equation}

\begin{figure}
    \centering
    \includegraphics[width=0.6\linewidth]{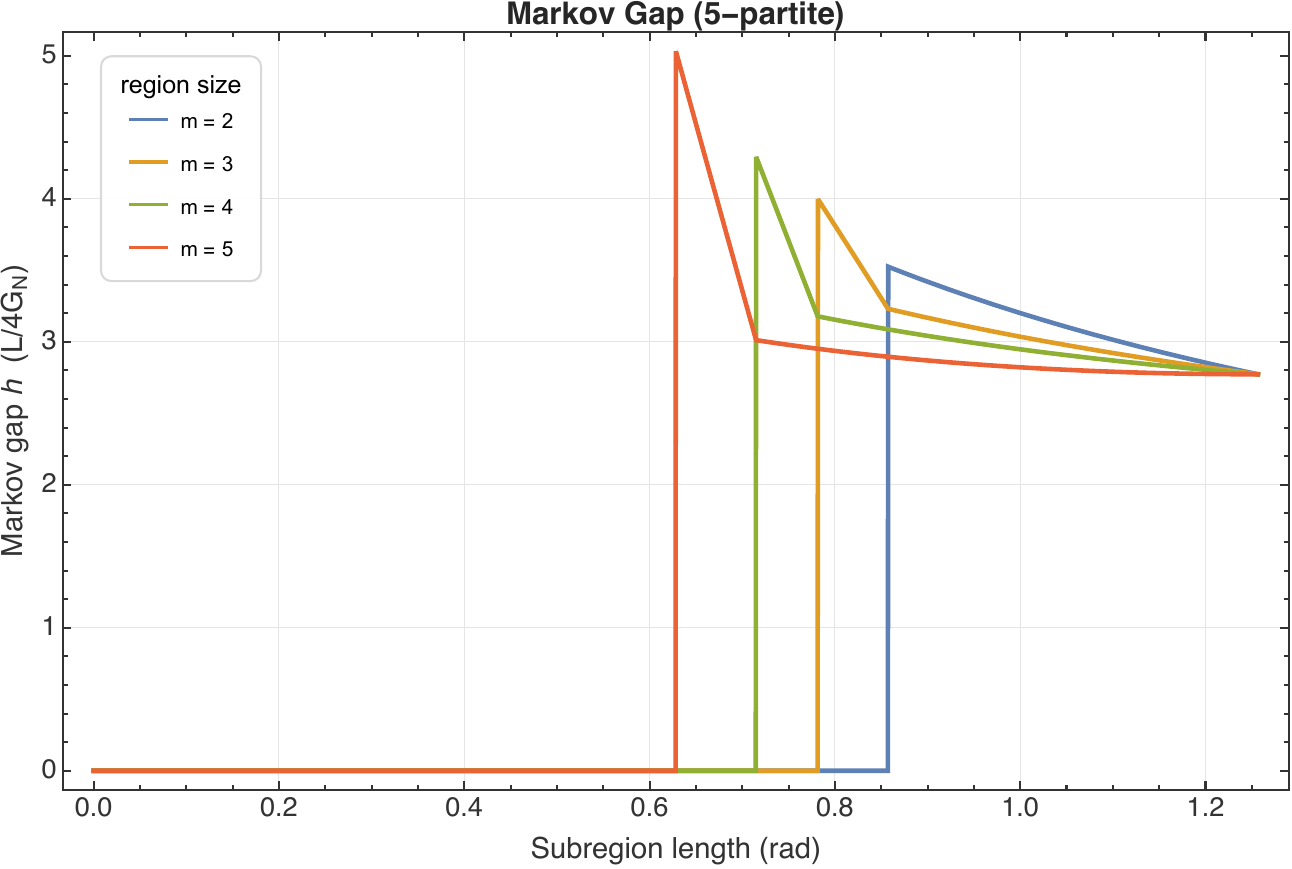}
    \caption{\small{The Markov gaps of adjacent unions of $m=2,\cdots,5$ boundary subregions in the symmetric $n=5$ setup, as functions of the angular interval $\phi$. The transition points coincide with the transition points in figure \ref{fig:transition_geometry_n-5-221}-(a).}}
    \label{fig:markov_gap}
\end{figure}




We have also numerically computed the Markov gaps in the symmetric setup, see figure \ref{fig:markov_gap}. The Markov gap is defined as the difference between the reflected entropy and the mutual information \cite{Hayden:2021gno}, i.e.,
\begin{equation}
    S_R -I.
\end{equation}
Here, for two boundary subregions $A$ and $B$, the mutual information is
\begin{equation}
    I(A:B) = S_A + S_B - S_{AB},
\end{equation}
and the reflected entropy $S_R(A:B)$ is the entanglement entropy of $AA^*$ in the canonical purification $\ket{\sqrt{\rho_{AB}}}$ of $\rho_{AB}$, where $A^*B^*$ denotes the copy of $AB$ \cite{Dutta:2019gen}. In figure \ref{fig:markov_gap}, we compute the Markov gaps with the bipartitions of $A_1:A_2$, $A_1:A_2A_3$, $A_1:A_2A_3A_4$, $A_1:A_2A_3A_4A_5$, and $A_1:A_2A_3A_4A_5A_6$. Holographically, the reflected entropy is computed by twice the area of the entanglement wedge cross section of $EW(AB)$ \cite{Dutta:2019gen}. Note that, for time-symmetric states in pure AdS$_3$, the Markov gap is universally bounded from below by $\frac{\log2}{2G_N}$ times the number of endpoints of the entanglement wedge cross section, in units of the AdS radius \cite{Hayden:2021gno} which is seen in the limiting behavior of figure~\ref{fig:markov_gap}.

The transition points extracted from the numerical computations of these entropic quantities are the same\footnote{They match at least at the order of four significant figures.} as the transition points computed above, such as figure \ref{fig:transition_points_n-3-4} and \ref{fig:angular_interval-n-5}. This agreement is expected because these quantities are computed from the same minimal surface configurations whose exchange defines the transitions, and it serves as a consistency check of our numerics.

\section{Redundancy in holographic encoding}\label{sec:redundant_encoding}

As we briefly discussed in the introduction, for $[n=3]=\{A_1,A_2,A_3\}$, let us consider two holographic phases when i) all bipartite and tripartite entanglement wedges are connected, figure \ref{fig:geometry_n-3}-(c), and ii) only the tripartite entanglement wedge is connected, figure \ref{fig:geometry_n-3}-(b).

In the first phase, such as figure \ref{fig:geometry_n-3}-(c), the bulk subregion $b$, which is inside all the entanglement wedges, is well protected against the erasure of any single boundary subregion. Even more, the bulk logical information in $b$ can be reconstructed from any pair of boundary subregions. This bulk subregion appears when the bipartite entanglement wedges are connected, see figure \ref{fig:mixed_error_correction_volume}\footnote{In this paper, we characterize the spatial bulk subregions in $AdS_3$ by the reconstructibility. However, it could be interesting to characterize them by volume. We comment on this perspective in section \ref{sec:discussions} as a future direction.}.

\begin{figure}
    \centering
    \includegraphics[width=0.9\linewidth]{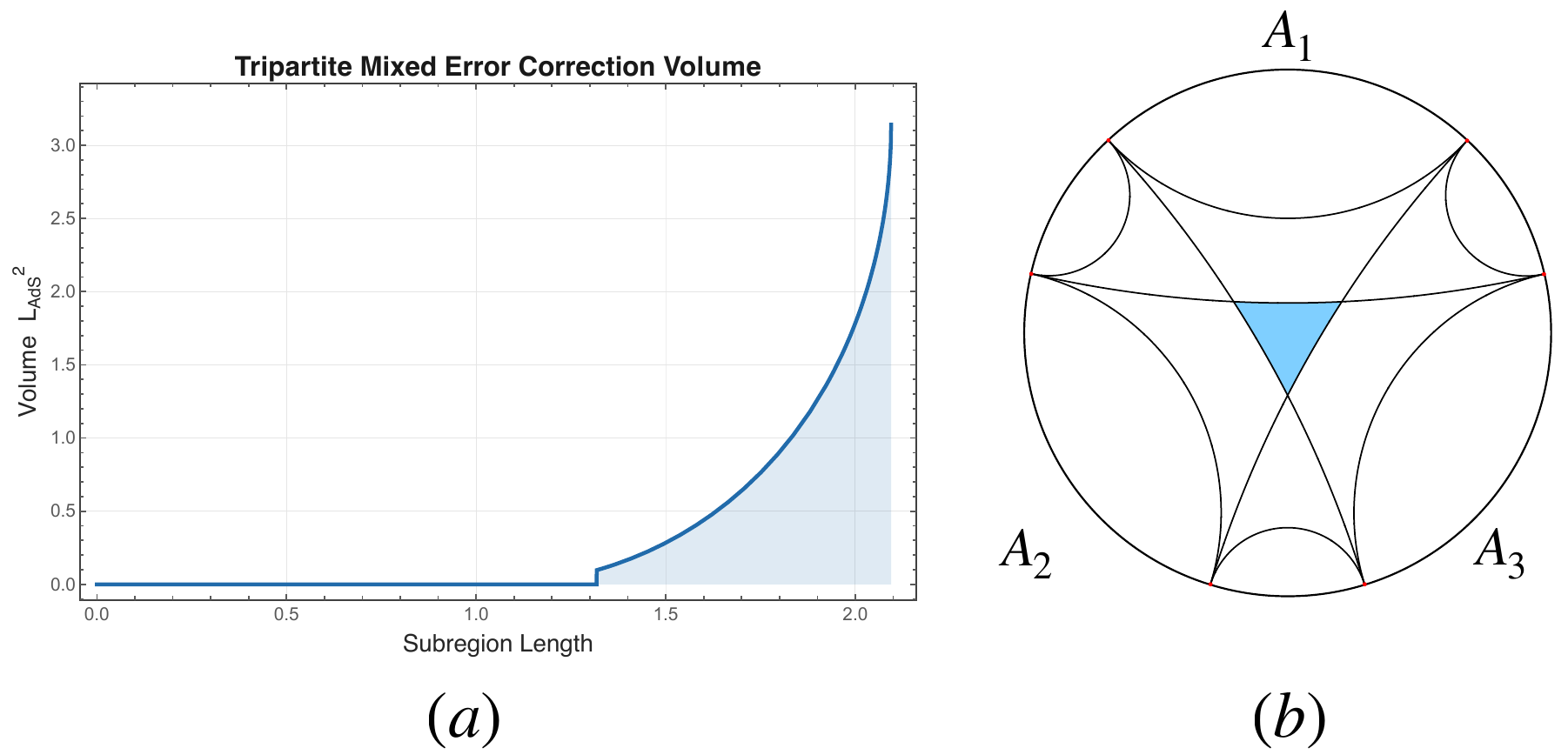}
   \caption{\small{(a) The volume of the bulk subregion contained in every bipartite entanglement wedge in the symmetric $n=3$ setup, as a function of the size of the boundary subregions. The volume vanishes until the bipartite entanglement wedges become connected and grows monotonically thereafter. (b) The corresponding bulk subregion, highlighted blue, bounded by the RT surfaces of the bipartite entanglement wedges.}}
    \label{fig:mixed_error_correction_volume}
\end{figure}

In the second phase, such as figure \ref{fig:geometry_n-3}-(b), the bulk subregion $b$ lies only within the tripartite entanglement wedge and is outside of the entanglement wedges of the single boundary subregions. Hence, we need three boundary subregions to reconstruct. Moreover, it cannot be reconstructed after the erasure of any single boundary subregion. 

In this section, we introduce a distance in the combinatorial setup to study how well a bulk subregion is protected against erasure\footnote{The distance introduced below can be considered as a coarse-grained version of the distance introduced in \cite{Pastawski:2016qrs}.}. However, we study it only in the coarse-grained setup, which we now define.

For a set $[n]:=\{A_1,\cdots, A_n\}$ of $n$ boundary subregions, consider a certain phase of entanglement wedges of $\mP[n]$. We coarse-grain the holographic geometry by introducing the RT-region graphs.
\begin{definition}[RT-region Graph\cite{Bao:2015bfa,Bao:2024azn,Bao:2025sjn}]\label{def:RTtoPC}
    Consider the RT surfaces of all elements of $\mP[n]$ in AdS$_{D+1}$/CFT$_D$ in a given holographic phase, which partition the bulk constant-time slice into bulk subregions. The RT-region graph  $G=(V,E)$ has the following elements.
    \begin{itemize}
        \item Set $V$ of vertices: Assign a vertex to every bulk subregion
            \begin{itemize}
                \item Boundary vertices $\{v_{A_i}\}\subseteq V$: A vertex assigned to the bulk subregion homologous to the boundary region $A_i\in [n]$.
                \item Bulk vertices $\{v_b\}\subset V$; the vertices that are not boundary vertices.
            \end{itemize}
        \item  Set $E$ of edges: Assign an edge between two vertices if the geodesics between them crosses an RT surface once. 
        \begin{itemize}
            \item The edge weights are proportional to the portion of the area of RT surfaces where the edges cross.
        \end{itemize}
    \end{itemize}
    
\end{definition}

We will use the graphs for visualizations, e.g., figure \ref{fig:graphs_n-3_distance_disconnected_partial_full}, and the definition of combinatorial holographic QSS in section \ref{sec:CHQSS}. Note that we often use $b$ instead of $v_b$ to represent a vertex in an RT-region graph.







\subsection{Setup and definitions}


Consider a bulk subregion $b$ and $n$ boundary subregions $[n]:= \{A_1, \cdots, A_n\}$ whose union may or may not cover the entire boundary\footnote{We denote by $O$ the complement of $\cup_{i=1}^n A_i$.}. \textit{Entanglement wedge reconstruction} \cite{Jafferis:2015del, Dong:2016eik,Harlow:2016vwg,Pastawski:2015qua,Almheiri:2014lwa,Furuya:2020tzv,Leutheusser:2022bgi,Cotler:2017erl,Chen:2019gbt} states that logical information in the bulk subregion $b$ can be reconstructed on the boundary subregion $R$ if 
\begin{equation}
    b\subseteq EW(R).
\end{equation}

Based on the entanglement wedge reconstruction, we define a set $\mathscr{R}_b$ of boundary subregions whose entanglement wedge contains the bulk subregion $b$, i.e.,
\begin{equation}\label{def:reconstructible_wedges_1}
    \mathscr{R}_{b}: = \Big\{R\in \mP[n]\; \big|\; b\subseteq EW(R) \Big\},
\end{equation}
where $\mP[n]$ is a power set of $[n]$.

The set $\msR_b$ enjoys an upward monotonicity property because of the \textit{entanglement wedge nesting}(EWN)\cite{Dong:2016eik,Akers:2017ttv,Czech2025EW,Czech:2023xed,Akers:2023fqr,Bao2025entanglementwedge}. The entanglement wedge nesting states that $R_1 \subseteq R_2$ implies $EW(R_1) \subseteq EW(R_2)$. If $b$ can be reconstructed from $R\in \msR_b$, then, any element $R' \in \mP[n]$ such that $R\subseteq R'$ should also be able to reconstruct $b$ because $R\subseteq R'$ implies that 
\begin{equation}
     b\subseteq EW(R) \subseteq EW(R')
\end{equation}
from the EWN. Thus, $R'\in \msR_b$.
\begin{lemma}[Upward monotonicity]\label{lem:upward_monotonicity}
    For $R\in \msR_b$ and any $R'\in\mP[n]$,
    \begin{equation}
        R\subseteq R' 
    \end{equation}
    implies that $R'\in \msR_b$.
\end{lemma}

We can identify the set of minimal elements that together generate $\msR_b$, i.e.,
\begin{equation}
    \msR_b^{\min} := \{R\in \msR_b \; \big|\; \nexists R'\in \msR_b, R' \subset R \}.
\end{equation}
Hence, the definition of $\msR_b$ in \eqref{def:reconstructible_wedges_1} can be restated as
\begin{equation}\label{def:reconstructible_wedges_2}
    \mathscr{R}_b = \Big\{R'\in \mP[n]\; \big|\; \exists R \in  \msR_b^{\min}, R \subseteq R' \Big\}.
\end{equation}

In addition to the EWN, we assume the complementarity of entanglement wedges for the entire boundary. That is, the entanglement wedge of the complement $R^c$ of the boundary subregion $R$ including the purifier $O$ is the complement $EW(R)^c$ of $EW(R)$, i.e.,
\begin{equation}
    EW(R^c) = EW(R)^c.
\end{equation}

We now establish the following lemma.
\begin{lemma}[Pair-wise intersections]\label{lem:pair-wise_intersection}
    Any pair of elements $R,R' \in \mathscr{R}_b^{\min}$ has non-trivial intersection, i.e.,
    \begin{equation}\label{eq:pair-wise_intersection}
        R\cap R' \neq \emptyset,\; \forall R,R'\in \mathscr{R}_b^{\min}
    \end{equation}
\end{lemma}
\begin{proof}
    Suppose two disjoint elements $R,R'\in \msR_b^{\min}$, i.e., $R\cap R'=\emptyset$, can reconstruct a bulk subregion $b$ each individually. It then should be that $b\subseteq EW(R)$ and  $b\subseteq EW(R')$, and hence, $b\subseteq EW(R)\cap EW(R')$. However, it contradicts that $EW(R)\cap EW(R') = \emptyset$ by the EWN and the complementarity. 
    
    Since $R\cap R'=\emptyset$, the geometric complement $R^c$\footnote{In this paper, we denote by $R^c$ and $EW(R)^c$ the geometric complement of spatial regions $R$ and $EW(R)$ with respect to the entire boundary or bulk spacetime. We denote by $\overline{R}$ the set-theoretic complement with respect to the set $[n]$.\label{ftnt:complement}} contains $R'$, i.e., $R^c\supset R'$. From the complementarity, we have $EW(R^c)$, which intersects $EW(R)$ only on the RT surface of $R$.  The EWN implies $EW(R^c) \supseteq EW(R')$ from $R^c\supseteq R'$. Hence, if $R\cap R'=\emptyset$, then, $EW(R)\cap EW(R')=\emptyset$.

    By definition, $b$ should be constructed from every element in $\msR_b^{\min}$, i.e., $EW(R)\cap EW(R')\neq\emptyset$ for any pair of $R,R'\in \msR_b^{\min}$. Therefore, \eqref{eq:pair-wise_intersection} holds.

\end{proof}



Here, we define how an erasure of boundary subregions applies to $\msR_b$. For a subset $\mE_E \subseteq [n]$, a \textit{set erasure} $\setminus_E$\footnote{The symbol $\setminus$ represents a set minus. The symbol $\setminus_E$ represents a specific set operation defined in \eqref{eq:set_erasure}. } of $\mE_E $ on $\msR_b$ is defined as
\begin{equation}\label{eq:set_erasure}
    \msR_b \setminus_E \mE_E  : = \msR_b \setminus \{R\; \big|\; R \cap \mE_E\neq \emptyset \} = \{R\in\msR_b\; \big|\; R \cap \mE_E =  \emptyset \} .
\end{equation}
That is, the erasure of $\mE_E$ removes from $\msR_b$ every combination of boundary subregions that contain at least one erased boundary subregion.

We now introduce the redundant encoding below.
\begin{definition}[Redundant encoding]\label{def:redundant_encoding}
    Consider a set of $n$ boundary subregions $[n]:=\{A_1, \cdots, A_n\}$, and a subset $\mE_E  \subseteq [n]$ of boundary subregions to be erased. A bulk subregion $b$ is redundantly encoded to subsets of $[n]$ against the erasure of $\mE_E $ if there is at least a single remaining union of boundary subregions in $[n]$, whose entanglement wedge contains the bulk subregion $b$, or 
    \begin{equation}
        |\mathscr{R}_b\setminus_E \mE_E |\geq 1,
    \end{equation}
    where $|\cdot|$ is the cardinality of a set\footnote{In this paper, $|\cdot|$ represents the cardinality of a set.}.
\end{definition}

\subsection{Distance}

We introduce a distance $d_b$ for a choice of bulk subregion $b$ to characterize the quality of the holographic encoding. The distance $d_b$ is the minimum size of erasures $\mE_E$ for which $b$ cannot be reconstructed\footnote{This definition of the distance is a direct analogue of the distance in the study of quantum error correction codes.}. Thus, there could be $\mE_E$ whose size is bigger than $d_b$, yet $b$ is still reconstructible. 

\begin{definition}[Distance]
    A distance on $\msR_b$ is defined as 
    \begin{equation}\label{eq:worst_case_distance}
    d_b := \underset{\mE_E \in [n]}{\min}|\mE_E | 
\end{equation} 
such that
\begin{equation}
    |\msR_b\setminus_E \mE_E | = 0.
\end{equation}
\end{definition}

\begin{figure}[t]
    \centering
    \includegraphics[width=1.0\linewidth]{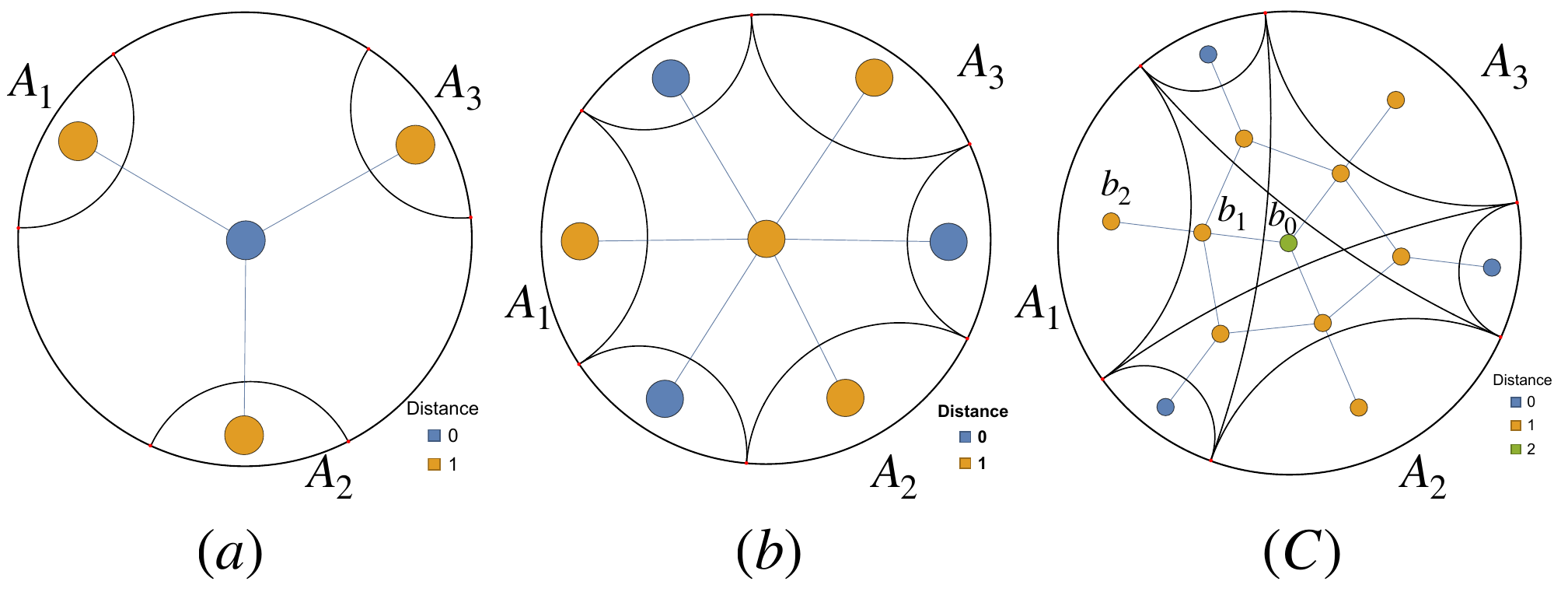}
    \caption{\small{The region graphs whose vertices are colored based on the values of the distance $d_b$: (a) Disconnected phase ($\phi<\phi_{t=0}$), (b) Only tripartite entanglement wedge is connected ($\phi_{t=1} \leq \phi<\phi_{t=2}$), (c) All possible entanglement wedges are connected($\phi\geq\phi_{t=2}$). }}
    \label{fig:graphs_n-3_distance_disconnected_partial_full}
\end{figure}

For example, in figure \ref{fig:graphs_n-3_distance_disconnected_partial_full}-(a), (b), and (c), we construct the region graphs for the three distinct holographic phases. Their vertices have different colors based on the distance.

For a choice of $n$, the maximum distance is determined as follows.
\begin{theorem}[Maximum distance, $d_{max}$]\label{prop:max_distance}
    In pure AdS$_3$, for fixed $n$, suppose there is a holographic phase such that entanglement wedges of all $\lceil \frac{n+1}{2} \rceil$ boundary subregions are connected, the maximum distance is given by 
    \begin{equation}
        d_{max}  := \underset{b}{\max} \; d_b =  \Big\lceil \frac{n}{2} \Big\rceil.
    \end{equation}
\end{theorem}
\begin{proof}
    Appendix \ref{app:proof_max_distance}
\end{proof}
We can see in figure \ref{fig:graphs_n-3_distance_disconnected_partial_full} that the distance $d_b$ increases as $b$ gets further from the boundary.

\section{Combinatorial holographic quantum secret sharing}\label{sec:CHQSS}


We study how the logical information stored in a bulk subregion can be encoded to, or reconstructed from, boundary subregions. Instead of the holographic QEC\cite{Almheiri:2014lwa,Harlow:2016vwg,Furuya:2020tzv,Pastawski:2015qua,Leutheusser2025quantumtasksQEC,Faulkner:2020hzi}, we study the question within the framework of quantum secret sharing (QSS). We construct a family of \textit{combinatorial holographic quantum secret sharing} (CHQSS) schemes, each defined for a bulk subregion. In particular, consider a set $[n]$ of $n$ boundary subregions and the entanglement wedges of the combinations of boundary subregions in $\mP[n]$ at a certain holographic phase. The set of RT surfaces corresponding to the entanglement wedges partitions a constant-time slice of AdS$_3$. A secret, or the logical information in the bulk subregion $b$, is distributed to $n$ players, or $n$ boundary subregions. The access and forbidden structures are constructed based on the entanglement wedge reconstruction. We introduce the reconstruction threshold $r_b$ and the secret threshold $s_b$ for $b$. We write as $(\!(r_b,s_b,n)\!)$-CHQSS the holographic quantum secret sharing of the bulk subregion $b$. 


At the end of this section, we present some examples using the symmetric setup studied in section \ref{sec:transition_points_MEW}. 




\subsection{Setup and definitions}


The access structure of $b$ consists of authorized sets. Those sets are the combinations of boundary subregions whose entanglement wedges contain $b$. Hence, it is the set $\msR_b$ defined in \eqref{def:reconstructible_wedges_1} and \eqref{def:reconstructible_wedges_2}. Let us remind ourselves of the definition of $\msR_b$ here, which is
\begin{equation}
    \mathscr{R}_{b}: = \Big\{R\in \mP[n]\; \big|\; b\subseteq EW(R) \Big\}.
\end{equation}
As in lemma \ref{lem:upward_monotonicity}, $\msR_b$ satisfies the monotonicity property, and is generated by the set of minimal elements $\msR_b^{\min}$.

On the contrary, the forbidden structure of $b$ has unauthorized sets, which are the combinations of boundary subregions whose entanglement wedges do not contain $b$, and thus cannot reconstruct it. Thus, it is defined as
\begin{equation}
    \msR_b^*: = \Big\{R \in \mP[n]\; \big|\; b\not\subseteq EW(R)\Big\}.
\end{equation}
One can simply write $\msR_b^*$ as the complement of $\msR_b$ with respect to $\mP[n]$, i.e.,
\begin{equation}
    \msR_b^* = \mP[n] \setminus \msR_b.
\end{equation}
Unlike the access structure, the forbidden structure\footnote{The forbidden structure is an abstract simplicial complex. See, for instance, \cite{volic2025topology}, for homological aspects of secret sharing.} satisfies downward monotonicity and contains maximal elements. 
\begin{lemma}[Downward monotonicity]\label{lem:downward_monotonicity}
    For $R\in \msR_b^*$ and any $R'\in\mP[n]$,
    \begin{equation}
        R'\subseteq R 
    \end{equation}
    implies that $R'\in \msR_b^*$.
\end{lemma}
The set of maximal elements is defined by
\begin{equation}
    \msR_b^{*\max} := \{R \; \big|\; \nexists R'\in \msR_b^*, R \subset R' \}.
\end{equation}

    

\begin{definition}[Combinatorial holographic quantum secret sharing (CHQSS)]
    Consider the RT-region graph for a set $[n]$ of $n$ boundary subregions, which may or may not cover the entire boundary, and their entanglement wedges of $\mP[n]$ in pure AdS$_{D+1}$. Combinatorial holographic secret sharing (CHQSS) for a bulk subregion $b$, or a vertex $v_b$, is an encoding scheme with an access structure $\msR_b$, which satisfies the following.
    \begin{enumerate}
        \item The entanglement wedge reconstruction: Logical information in the bulk subregion $b$ can be reconstructed on the boundary subregion $R$ if 
        \begin{equation}
            b\subseteq EW(R).
        \end{equation}
        \item The entanglement wedge nesting: $R_1 \subseteq R_2$ implies $EW(R_1) \subseteq EW(R_2)$ for boundary subregions $R_1,R_2$.
        \item The geometric complementarity of entanglement wedges: The entanglement wedge of the complement $R^c$ of the boundary subregion $R$ including the purifier $O$ is the complement $EW(R)^c$ of $EW(R)$\footnote{As mentioned in footnote \ref{ftnt:complement}, we denote by $R^c$ and $EW(R)^c$ the geometric complement of spatial regions $R$ and $EW(R)$ with respect to the entire boundary or bulk 
        space. We denote by $\overline{R}$ the set-theoretic complement with respect to the set $[n]$.}, i.e.,
        \begin{equation}
        EW(R^c) = EW(R)^c.
        \end{equation}

    \end{enumerate}
\end{definition}
Note that we obtain a family of CHQSS from an RT-region graph if the graph has several vertices. 

\begin{remark}
    By lemma \ref{lem:upward_monotonicity}, the access structure $\msR_b$ is monotone, and by lemma \ref{lem:pair-wise_intersection}, no two disjoint elements of $\mP[n]$ are both authorized. By the result of \cite{Gottesman:1999jzr}, these two properties guarantee that a quantum secret sharing scheme with the access structure $\msR_b$ exists. 
\end{remark}
\begin{remark}
    We call a CHQSS scheme \textit{pure-state} if the union of the boundary subregion in $[n]$ covers the entire boundary, i.e., $O=\emptyset$, and \textit{mixed-state} otherwise. 
\end{remark}
\begin{remark}\label{rem:types_QSS}
    A general access structure for QSS can be characterized by the parameters $(\!(r,s,n)\!)$ and is classified in three ways. They are (i) threshold or non-threshold and (ii) perfect or ramp. The first classification is a \textit{threshold} or \textit{non-threshold} QSS. On the one hand, the threshold QSS is fully characterized by the parameters $r$, $s$, and $n$, and does not depend on the details of the authorized sets in the access structure. On the other hand, the non-threshold QSS depends on the details of the authorized sets. For the second one, a QSS is a ramp if there is an intermediate set in addition to the authorized and unauthorized sets. The intermediate sets contain partial information about the encoded information. In our holographic setup, we do not find a ramp QSS because the intermediate sets are empty due to the holographic phase transition. Consider, for example, the tripartite connected phase in figure \ref{fig:geometry_n-3}-(b). Two boundary subregions do not contain any bulk logical information in $b$. However, once all three boundary subregions are accessible, the bulk logical information can be reconstructed. This sharp transition hinders a boundary subregion from possessing partial bulk logical information.
\end{remark}

We now introduce two parameters, \textit{reconstruction threshold} $r_b$ and \textit{secret threshold} $s_b$ defined as follows.
\begin{definition}[Reconstruction threshold]
    Any combination of $r_b$ boundary subregions, or more, can reconstruct $b$, i.e.,
    \begin{equation}\label{eq:def_reconstruction_threshold}
        r_b = \underset{R\in X_b^R}{min} |R|
    \end{equation}
    where the minimization is over a set\footnote{Note that $X_b^R$ is in general a subset of $\msR_b$, i.e., $X_b^R \subseteq \msR_b$. The set $X_b^R$ can be understood as a set of boundary subregions that can uniformly reconstruct the bulk logical information in $b$. It is \textit{uniform} because any combinations of at least $|R|$ boundary subregions can reconstruct the bulk logical information. \label{fnt:uniform}},
    \begin{equation}
        X_b^R := \{R\in \mP[n]\;| \; b\subseteq EW(R'),\;\forall R'\in \mP[n] \text{ with } |R'|\geq |R|\}.
    \end{equation}
    $|R|$ is the number of boundary subregions composing $R$.
\end{definition}
The reconstruction threshold has a relation to the distance $d_b$. 
\begin{theorem}\label{thm:reconstruction_distance}
    
    For a bulk subregion $b$ with $b\subset EW([n])$,
    \begin{equation}\label{eq:reconstruction_distance}
    r_b =  n-d_b+1. 
\end{equation}    
\end{theorem}
\begin{proof}
    For any $R\in\mP[n]$, there is $\mE_E$ such that 
    \begin{equation}
        R = [n]\setminus \mE_E,
    \end{equation}
    and we have
    \begin{equation}\label{eq:b-in-R_b-out-E}
        b\subseteq EW(R) \Rightarrow b\not\subseteq EW(\mE_E).
    \end{equation}
    Then, from \eqref{eq:def_reconstruction_threshold},
    for a set of correctable erasures of boundary subregions\footnote{The erasure errors in the set $X_b^E$ are uniformly correctable because the bulk logical information in $b$ is uniformly reconstructable in the sense of footnote \footref{fnt:uniform}. },
    \begin{equation}\label{eq:X_b_def}
        X_b^E := \{\mE_E\in \mP[n] \;|\; b\not\subseteq EW(\mE'_E),\; b\subseteq EW(R'=[n]\setminus \mE'_E),\;\forall \mE'_E\in \mP[n] \text{ with } |\mE'_E|\leq |\mE_E|\},  
    \end{equation}
    we get 
    \begin{equation}
        r_b = \underset{\mE_E \in X_b^E}{\min} |[n]\setminus\mE_E| =n - \underset{\mE_E \in X_b^E}{\max} |\mE_E|.
    \end{equation}
    The second term is the maximum size of $\mE_E$ such that $b$ can still be reconstructed from $R=[n]\setminus \mE_E$. Hence, it equals $d_b-1$ and proves \eqref{eq:reconstruction_distance}.
    
\end{proof}

From theorem \ref{prop:max_distance}, we get the minimum of $r_b$, i.e.,
\begin{equation}\label{eq:r_min}
    r_{\min} =n- d_{\max} +1 =\Big\lfloor \frac{n}{2} \Big\rfloor +1.
\end{equation}

Note that $r_{\min}$ does not violate the bound on the reconstruction threshold given by the no-cloning theorem, i.e., it satisfies $2r_{\min}> n$ \cite{Gottesman:1999jzr,Cleve:1999qg}. That is, we have
\begin{equation}
    r_{\min} =  \Big\lfloor \frac{n}{2} \Big\rfloor +1 > \frac{n}{2}.
\end{equation}

The original derivation of the bound was as follows. If cloning an unknown quantum state is possible, then two disjoint elements of the access structure can individually reconstruct it. It should then be $ 2r \leq n$. However, the no-cloning theorem prohibits the copying of a quantum state. Thus, it should be $2r>n$.

In the holographic set, the bound can be proved as follows\footnote{Holographic versions of the no-cloning theorem has been discussed in \cite{Bousso:2025fgg,Bousso:2023sya,Bousso:2022hlz}.}. Suppose cloning is possible and $2r_b\leq n$. Then, there are two disjoint combinations $R,R'\in \msR_b^{\min}$ of size $|R|=|R'|=r_b$. Both $R$ and $R'$ should be able to reconstruct $b$ by the definition of $b$. However, this contradicts $EW(R)\cap EW(R') =\emptyset$ as discussed in the proof of lemma \ref{lem:pair-wise_intersection}. Thus, we have the bound $2r_b > n$.



\begin{definition}[Secret threshold]
    Any combination of $s_b$ boundary subregions, or less, cannot reconstruct $b$. That is,
    \begin{equation}
        s_b = \underset{R\in X_b^S}{max} |R|
    \end{equation}
    for a set\footnote{The set $X_b^S$ represents the uniform secrecy of the bulk logical information in $b$. That is, the bulk logical information in $b$ is kept secret from any combination of at most $|R|$ boundary subregions. We should also note that $X_b^E$ and $X_b^S$ are not the same in general. The set $X_b^E$ assumes $b\not\subseteq EW(\mE'_E)$ and $ b\subseteq EW(R'=[n]\setminus \mE'_E)$ whereas the set $X_b^S$ assumes $b\not\subseteq EW(\mE'_E)$ only.},
    \begin{equation}
        X_b^S:= \{R\in \mP[n]\; | \;b\not\subseteq EW(R'),\;\forall R'\in\mP[n]\text{ with } |R'|\leq |R|\}.
    \end{equation}
\end{definition}

The maximum of $s_b$ for $b\subseteq EW([n])$ is 
\begin{equation}
    s_{max} = n-1.
\end{equation}
This can be checked as follows. Suppose $b\subseteq EW(R)$ for some $R$ with $|R|=m$, and $b\not\subseteq EW(R')$ for every $R'$ with $|R'|\leq m-1$. Then, $s_b=m-1$ by its definition. As we observed in section \ref{sec:transition_points_MEW}, there is a phase where an $n$-partite entanglement wedge is connected, but $n'$-partite entanglement wedges for any $n'<n$ are not connected. Thus, $s_{max}=n-1$.

We need extra care with the relations between the secret threshold and the distance, as opposed to those for the reconstruction threshold and the distance, as in \eqref{eq:reconstruction_distance}, because of \textit{superadditivity} of entanglement wedges\footnote{In \cite{Faulkner:2020hzi,Leutheusser2025quantumtasksQEC}, it was pointed out and studied that the superadditivity of local bulk algebras conflicts with the framework of exact quantum error correction in holography. For future work, it would be interesting to construct an operator algebraic holographic QSS and study the role of the superadditivity of local bulk algebras.}. Entanglement wedges are said to be superadditive if 
\begin{equation}
    EW(R_1) \cup EW(R_2) \subset EW(R_{1}\cup R_{2})
\end{equation}
for boundary subregions $R_1,R_2$. If the entanglement wedges of the boundary subregions $R_1,R_2$ are superadditive, there is a bulk subregion such that
\begin{equation}
    b \not\subseteq EW(R_1) \cup EW(R_2)\text{ and } b \subseteq EW(R_{1}\cup R_{2}).
\end{equation}
For such a bulk subregion, it follows that
\begin{equation}\label{eq:notin_complement}
    b\not\subseteq EW(R_i) \text{ and } b\not\subseteq EW(\overline{R_i}=[n]\setminus R_i)
\end{equation}
for $i=1,2$. 

\begin{figure}
    \centering
    \includegraphics[width=0.9\linewidth]{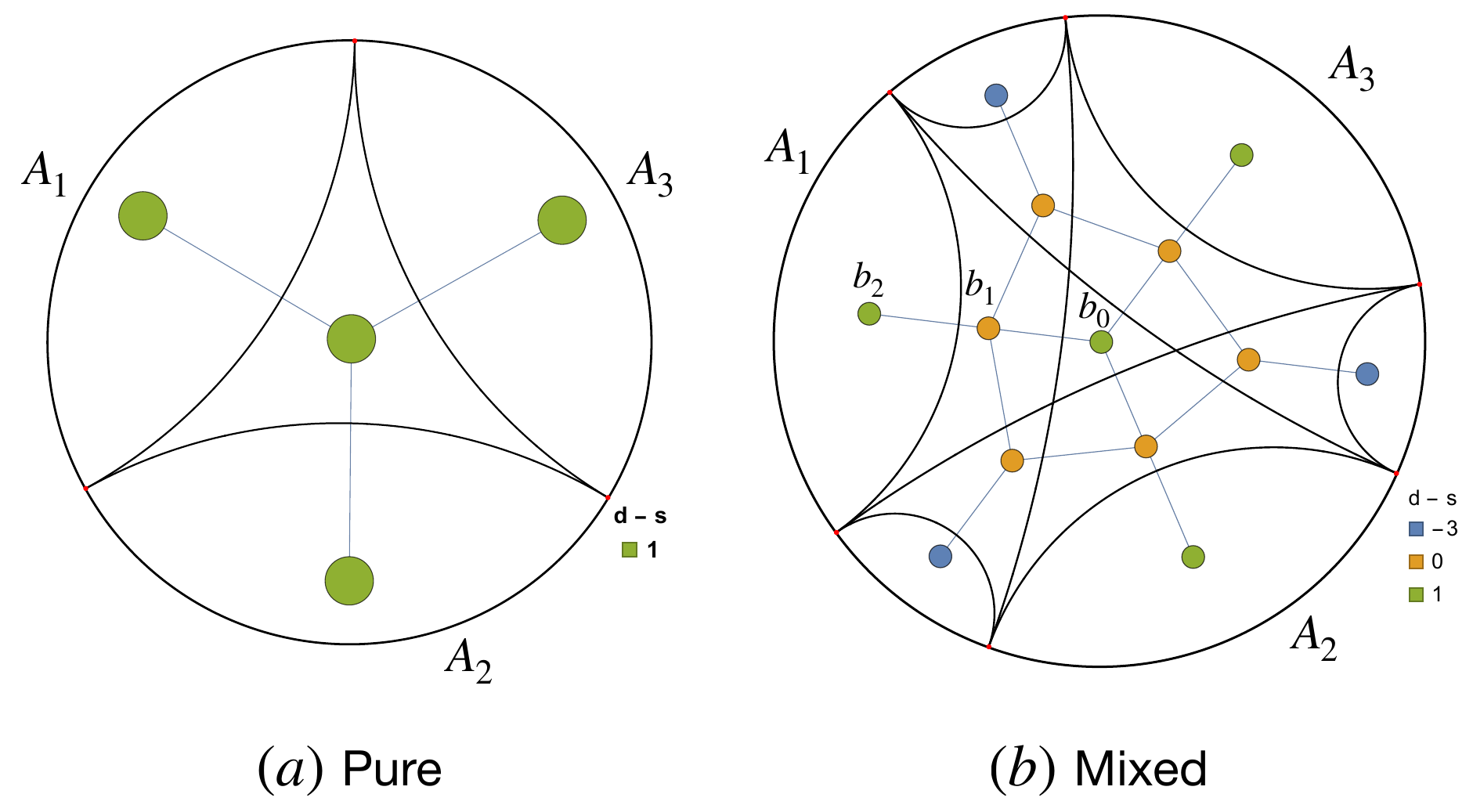}
    \caption{\small{The region graphs whose vertices are colored based on the value of the difference $d-s$: (a) A family of pure-state CHQSS schemes, (b) A family of mixed-state CHQSS schemes.} The bulk subregion represented as the vertex $b_1$ satisfies \eqref{eq:notin_complement}. See \eqref{eq:notin_complement_b1} for details.}
    \label{fig:graphs_n-3_d_s_difference_pure_mixed}
\end{figure}

We are especially interested in the following case, i.e.,
\begin{equation}\label{eq:notin_complement_pure}
    EW(R)\cup EW(\overline{R}) \subset EW([n]=R\cup \overline{R}).
\end{equation}
This can occur only in the mixed-state CHQSS, where the geometric complement $O$ of the $n$ boundary subregions in the entire boundary is not trivial, because, for $O=\emptyset$, the entanglement wedges  $EW(R)$ and $EW(\overline{R})$ are complementary and together cover the constant-time slice, see figure \ref{fig:graphs_n-3_d_s_difference_pure_mixed}. 

In the proof of \eqref{eq:reconstruction_distance}, we used the fact that $b\subseteq EW(R)$ implies that $b\not\subseteq EW(\overline{R})$, which holds for any bulk subregion. However, as in \eqref{eq:notin_complement}, $b\not\subseteq EW(R)$ is necessary but not sufficient for $b\subseteq EW(\overline{R})$. This distinction plays a key role in the following theorem.


\begin{theorem}\label{thm:secret_distance}

    Consider $b\subseteq EW([n])$, and write $\overline{R}:=[n]\setminus R$ for $R\in \mP[n]$. One of the following two statements holds.
    
    \begin{enumerate}
        \item Additive: $b\subseteq EW(\overline{R})$ for every $R\in \mP[n]$ with $|R|=s_b$ if and only if
        \begin{equation}\label{eq:secret_distance}
            s_b = d_b-1.
        \end{equation}
    
        \item Superadditive: $b\not\subseteq EW(R)$ and $b\not\subseteq EW(\overline{R})$ for some $R\in \mP[n]$ with $|R|=s_b$ if and only if
        \begin{equation}\label{eq:secret_distance_superadditive}
            s_b \geq d_b.
        \end{equation}
    \end{enumerate}

\end{theorem}
\begin{proof}
    Appendix \ref{app:proof_secret_distance}.
\end{proof}

\begin{corollary}\label{cor:pure-mixed}
    The pure-state CHQSS schemes are additive. The mixed-state CHQSS schemes are either additive or superadditive. 
\end{corollary}



If the CHQSS scheme is additive,
\begin{equation}
    r_b+s_b = n,
\end{equation}
whereas if it is superadditive,
\begin{equation}
    r_b+s_b \geq n+1.
\end{equation}

In a general QSS, the second equation implies a ramp CHQSS. However, in the CHQSS, we cannot have a ramp CHQSS because the intermediate sets are empty in the holographic setup as explained in remark \ref{rem:types_QSS}.


In the next subsection, we explore a family of pure- and mixed-state CHQSS in the symmetric cases. We will observe that CHQSS schemes can be either a perfect threshold QSS or a perfect non-threshold QSS. 






\subsection{Symmetric cases}

\begin{figure}[t]
    \centering
    \includegraphics[width=0.9\linewidth]{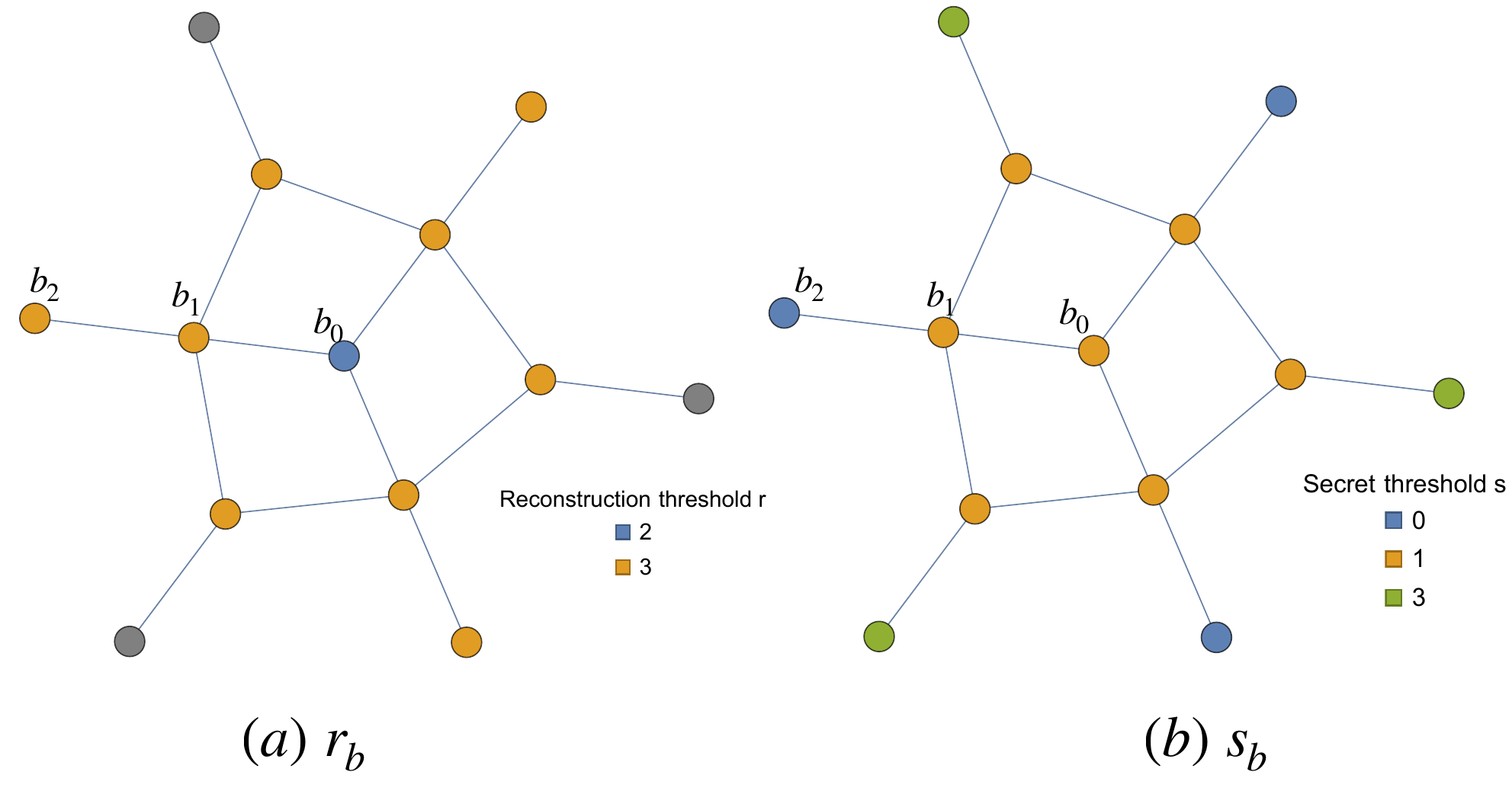}
    \caption{\small{The region graphs for $n=3$ boundary subregions in the holographic phase $\phi \geq \phi_{t=2}$ (a) The vertices are colored based on the value of the reconstruction threshold $r_b$. (b) The vertices are colored based on the value of the secret threshold $s_b$. }}
    \label{fig:graphs_n-3_r_s_threshold_full}
\end{figure}

\begin{figure}[t]
    \centering
    \includegraphics[width=0.9\linewidth]{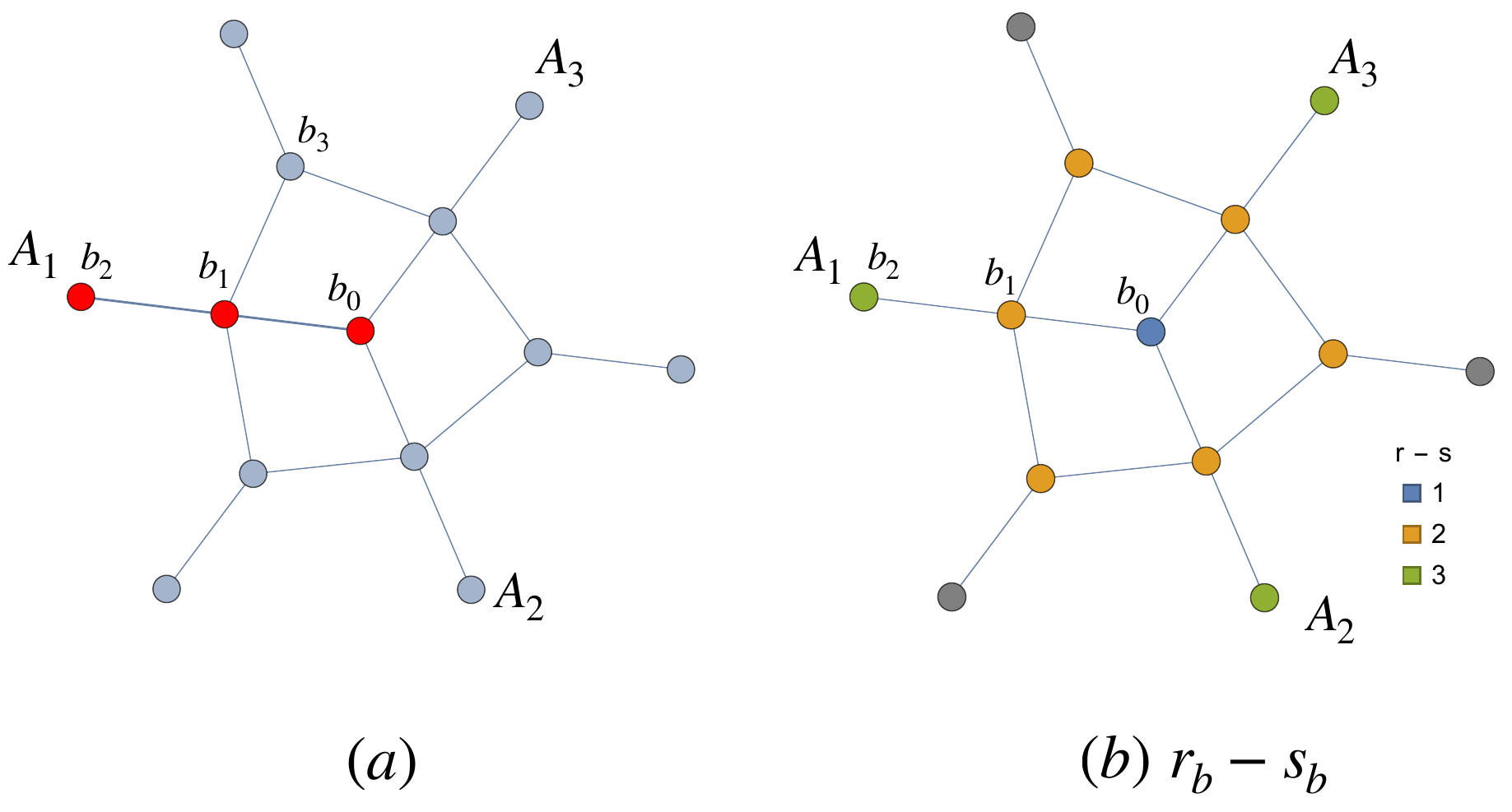}
    \caption{\small{The region graphs for $n=3$ boundary subregions in the holographic phase $\phi \geq \phi_{t=2}$ (a) The vertices $b_0$, $b_1$, and $b_2$ in the path from the center of the region graph to the boundary vertex $A_2$ are colored red. (b) The vertices are colored based on the value of the difference $r_b-s_b$. }}
    \label{fig:graphs_n-3_graph-center_ramp_full}
\end{figure}

In this subsection, we work out the example of $n=3$ boundary subregions in the holographic phase where the entanglement wedges of all elements of $\mP[n]$ are connected. We evaluate the properties of the family of mixed-state CHQSS for each bulk subregion $b$ by evaluating the distance $d_b$, the reconstruction threshold $r_b$, the secret threshold $s_b$, and some of their differences, such as $d_b-s_b$ and $r_b-s_b$.

In figure \ref{fig:graphs_n-3_distance_disconnected_partial_full}, \ref{fig:graphs_n-3_d_s_difference_pure_mixed}, \ref{fig:graphs_n-3_r_s_threshold_full}, and \ref{fig:graphs_n-3_graph-center_ramp_full}, we colored the vertices of the region graph based on the values of the parameters and the differences.

To study the properties of CHQSS at each vertex, we begin by defining the center of the region graph.
\begin{definition}[Center of region graph]
    For a region graph $G=(V,E)$ with a set of boundary vertices $v_{A_i}$ for $A_i \in [n]$, the center of the region graph, or the set $\mC_G$ of central vertices, with respect to the boundary subregions $A_i \in [n]$, is defined as\footnote{For a graph $G=(V,E)$, the quantity $\underset{u\in V}{\min}\;\underset{v\in V}{\max} \;d_G(u,v)$ is known as a \textit{graph radius}, see, for instance, \cite{Diestel2005graphtheory}.}
    \begin{equation}
        \mC_G : = \{ u\in V \;| \;\underset{v_{A_i}}{\max} \;d_G(u,v_{A_i}) =\underset{w\in V}{\min}\;\underset{v_{A_i}}{\max} \;d_G(w,v_{A_i})  \},
    \end{equation}
    where the maximization is over the boundary vertices, and $d_G$ is the graph distance, or the length of the shortest path between vertices.
\end{definition}

For example, let us consider four vertices $b_0$, $b_1$, $b_2$, and $b_3$, in figure \ref{fig:graphs_n-3_graph-center_ramp_full}-(a). The other vertices are related to these four by the symmetry of the graph. First, for each vertex $b_i$ with $i=0,1,2,3$, we compute the maximum graph distance to a boundary vertex, which is attained at the boundary vertex $v_{A_2}$, i.e.,
\begin{equation}
    d_G(b_0,v_{A_2}) = 2,\;d_G(b_1,v_{A_2}) = 3,\; d_G(b_2,v_{A_2}) = 4, \; d_G(b_3,v_{A_2}) = 4.
\end{equation}
Note that we have multiple boundary vertices that maximize the distance for each choice among the four vertices due to the symmetry of the graph. The vertex that minimizes the maximum distance is $b_0$. Hence, in this example, $b_0$ is the center of the region graph, i.e., $\mC_G=\{b_0\}$.

We choose a path from the center $b_0$ to a boundary vertex, such as $v_{A_1} = b_2$\footnote{One can choose the path between the center and $v_{A_2}$ or $v_{A_3}$. However, the results are independent of the choice because of the symmetry of the graph.}. In figure \ref{fig:plots_n-3_graph-center_d-r-s_differences}, we plot the values of the parameters $d_b$, $r_b$, $s_b$, $d_b-s_b$, and $r_b-s_b$, as a function of distance along the path from the center $b_0$.

The distance is maximum at the center, i.e., $d_{\max} = \big\lceil\frac{n}{2} \big\rceil= 2$. It monotonically decreases toward the boundary. The reconstruction threshold $r_b$ increases and the secret threshold $s_b$ decreases as the distance $d_b$ decreases from theorem \ref{thm:reconstruction_distance} and \ref{thm:secret_distance}, respectively.

The reconstruction threshold attains its maximum at the boundary vertices, except at the vertices homologous to the geometric complement $O$ of $[n]$\footnote{The vertices homologous to the geometric complement $O$ of $[n]$ can never be reconstructed from any combination of boundary subregions in $[n]$. Thus, their reconstruction threshold is undefined.}. Physically, we need $n$ boundary subregions to guarantee the reconstruction of any boundary vertex. For example, to reconstruct the vertex $b_2$ homologous to the boundary subregion $A_1$, every $R \in \msR_b^{\min}$ with $|R|\geq r_{b_2}$ must be able to reconstruct $b_2$. Since $EW(A_2A_3)$ does not contain $b_2$, this is possible only for $r_{b_2}=3$, i.e., $R = [n]$. 

Unlike the reconstruction threshold, the secret threshold reaches its minimum at the boundary vertices, i.e., $s_{\min} = d_{\min} - 1 = 0$.

We have a perfect threshold $(\!(r_{b_0},s_{b_0},n)\!)$-CHQSS at the center $b_0$. First, it is a perfect scheme because $r_b-s_b =1$. In addition, it is a threshold scheme because $\msR_{b_0}^{\min}$ contains all $R\in \mP[n]$ such that $|R| = r_{b_0}=2$, i.e.,
\begin{equation}
    \msR_{b_0}^{\min} = \{A_1A_2,A_2A_3,A_1A_3\}.
\end{equation}

As opposed to CHQSS at $b_0$, $(\!(r_{b_i},s_{b_i},n)\!)$-CHQSS for $i=1,2$ are perfect non-threshold schemes because whether a combination of boundary subregions can reconstruct $b_i$ depends on which boundary subregions it contains, not only its cardinality. That is,
\begin{equation}
    \msR_{b_1}^{\min} = \{A_1A_2, A_1A_3\},
\end{equation}
which does not contain $A_2A_3$ even though $|A_2A_3|=|A_1A_2|$, and 
\begin{equation}
    \msR_{b_2}^{\min} =  \{A_1\},
\end{equation}
which does not contain $A_2,A_3$.

We can see that $(\!( r_{b_1},s_{b_1},n)\!)$-CHQSS is superadditive because $s_{b_1}=d_{b_1}$ in figure \ref{fig:plots_n-3_graph-center_d-r-s_differences}. One can confirm in figure \ref{fig:graphs_n-3_d_s_difference_pure_mixed}-(b) that
\begin{equation}
    b_1 \not\subseteq EW(A_1) \cup EW(A_2A_3)\text{ and } b_1 \subseteq EW(A_1A_2A_3),
\end{equation}
or, equivalently, for $b_1 \subseteq EW(A_1A_2A_3)$,
\begin{equation}\label{eq:notin_complement_b1}
    b_1\not\subseteq EW(A_1) \text{ and } b_1\not\subseteq EW(\overline{A_1}).
\end{equation}




\begin{figure}[t]
    \centering
    \includegraphics[width=0.9\linewidth]{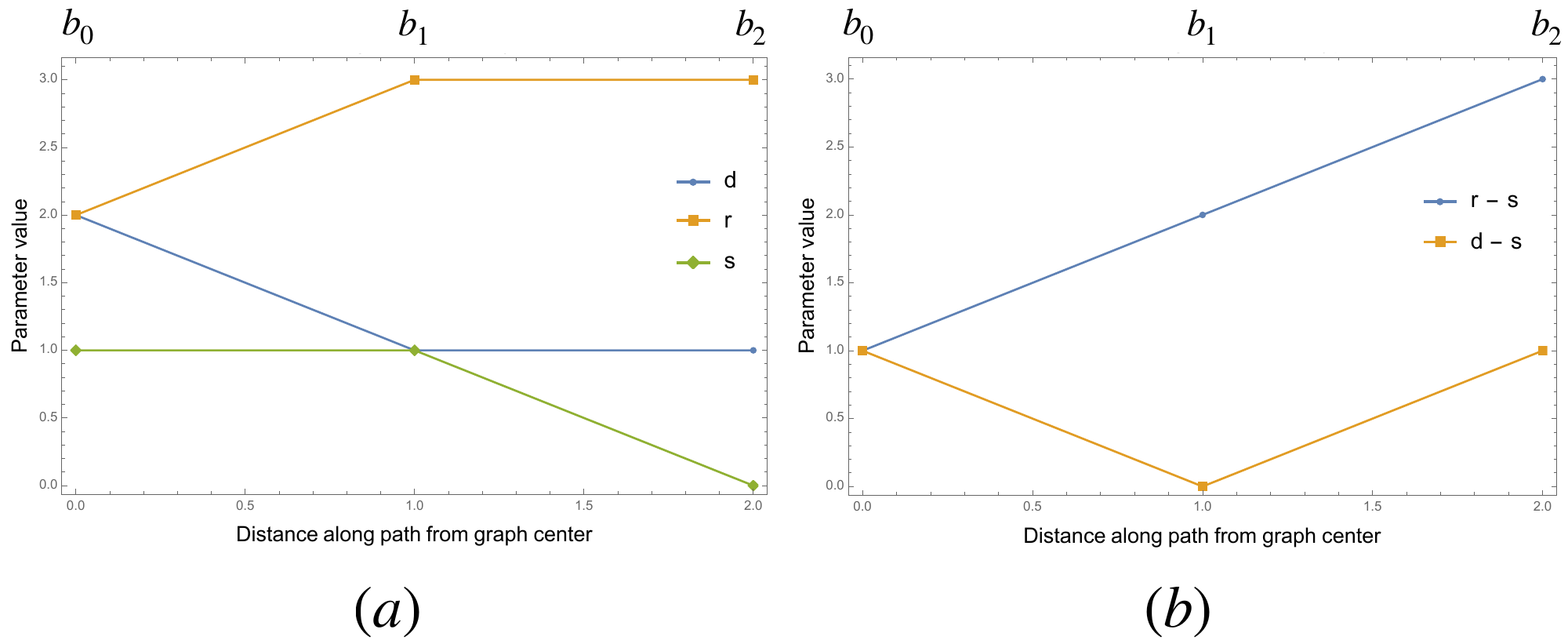}
    \caption{\small{The plots of the parameters $d_b$, $r_b$, $s_b$, and the differences $r_b-s_b$ and $d_b-s_b$ as a function of the distance along the path from the center of the graph: (a) $d_b$, $r_b$, and $s_b$. (b) $r_b-s_b$, and $d_b-s_b$.}}
    \label{fig:plots_n-3_graph-center_d-r-s_differences}
\end{figure}








\section{Discussions}\label{sec:discussions}


We studied the quality of holographic encoding by constructing combinatorial holographic quantum secret-sharing (CHQSS) whose access structure is characterized by the reconstruction threshold $r_b$ and the secret threshold $s_b$. Moreover, we proved the relation between the distance $d_b$ and the thresholds. 

In general, for a better QSS, we want a smaller $r_b$ for easier reconstruction, a larger $s_b$ for higher secrecy, and a larger $d_b$ for greater robustness. We found that any pure-state CHQSS scheme satisfies $d_b-s_b = 1$, whereas the superadditive mixed-state CHQSS schemes satisfy $s_b\geq d_b$. Thus, those mixed-state CHQSS schemes could achieve higher secrecy, but the encoded bulk logical information can be fragile to erasures of boundary subregions. It is intuitive because the superadditive schemes encode bulk logical information to the boundary via genuine multipartite entanglement wedges, see, for instance, Figure \ref{fig:geometry_n-3}-(b). 

In addition, in the mixed-state case, there is a trade-off between distance and the additivity. In some phases, the central vertex of the region graphs has the best distance among the other CHQSS schemes, but it does not exhibit additivity. One can recover the additivity by going away from the center, but the distance decreases monotonically as one moves away from the central vertex. Therefore, the best protected qubit that satisifies the additivity property is not necessarily at the center of $AdS$ space.


The additivity was essential to the exact HQEC, and the existence of the superadditive mixed-state CHQSS underlies the failure of the exact holographic QEC \cite{Faulkner:2020hzi,Leutheusser2025quantumtasksQEC} and is related to constraints on holographic QEC imposed by holographic quantum entropy inequalities \cite{Czech:2025jnw}. 
We now elaborate potential future directions below.

\subsection{Volume of codimension-one surface and its quantum information interpretations}

In this paper, we characterized the bulk subregions by the inclusion relations of entanglement wedges. One could instead characterize the bulk subregions by their volumes. For example, the recent paper \cite{Fujiki:2026ucr} introduced the entanglement polygon, a codimension-one surface, and studied the properties of its volume in various spacetimes. 

The operational interpretations of the volume of a bulk subregion have been explored in terms of complexity\cite{Alishahiha:2015rta,Roy:2017kha,Carmi:2016wjl,Ben-Ami:2016qex,Caceres:2019pgf,Caceres:2018blh,Bakhshaei:2017qud,Caceres:2018luq,Alishahiha:2018lfv,Agon:2018zso,Nakajima:2026agi,Gerbershagen:2024qlz}. A closely related quantity is the fidelity, a distinguishability measure in quantum information, which plays a key role in the construction of the Petz recovery map\cite{Junge:2015lmb,Denes:1986petzmap,Denes:2003petrecovery,Andras:2022petzrecovery,Knill:2002petzrecovery,Lautenbacher:2023sfa,Faulkner:2020iou,Faulkner:2020kit,Furuya:2020tzv} in the study of quantum error correction codes. It is interesting to explore operational interpretations of the volumes of the bulk subregions via holographic fidelity, holographic quantum error correction, and holographic quantum secret sharing.

In the operator algebraic approach, it was proposed that the volume of a bulk subregion can be related to the index of an inclusion of local boundary subalgebras \cite{Leutheusser:2025zvp}. It is also intriguing to study non-perturbative aspects of volumes of bulk subregions in relation to operator algebraic quantum error correction.

\subsection{Multipartite entanglement patterns from holographic QSS}


Entanglement is a resource for realizing quantum secret sharing. Entanglement structures of a quantum state affect an access structure, and thus determine the quality of QSS. For example, absolutely maximally entangled states can be used to construct perfect threshold QSS schemes\cite{Helwig:2012nha}. AME states or perfect tensors have been used to construct tensor-network toy models of holography \cite{Goyeneche:2015fda,Pastawski:2015qua}.

Conversely, a QSS scheme imposes constraints on the entanglement patterns of the quantum state associated with the QSS. This motivates studying the constraints on holographic states and whether they can be read off from the access structure of the holographic QSS schemes, including our combinatorial holographic QSS.





\subsection{Classification of quantum secret sharing}


Quantum secret sharing schemes depend on the class of quantum states. 
One can have stabilizer-state QSS\cite{Matsumoto:2017xuk,Matsumoto:2019wct} and graph-state QSS\cite{Markham:2008kqi,Keet:2010rlc,Sarvepalli:2012nfk}. Holographic quantum states form a restricted subclass of general quantum states, constrained, for example, by the holographic entropy inequalities (HEIs). 

One potential future direction is to study how the HEIs constrain the holographic QSS, and vice versa. In \cite{Czech:2025jnw}, the authors studied the constraints imposed on holographic quantum error correction by the HEIs.

The potential connection between holographic QSS and HEIs could be made through a quantum information rate. The quantum information rate quantifies the performance of quantum secret sharing \cite{Sarvepalli:2010qvu,Imai:2003fuj}. In particular, the paper \cite{Imai:2003fuj} provided an entropic characterization of QSS and an entropic version of the quantum information rate. It is interesting to explore how the HEIs could constrain and characterize the performance of holographic QSS, for example. 

Furthermore, combinatorial aspects of HEIs have been explored in \cite{Grimaldi:2026lbq}. It is also intriguing to explore if their majorization structure can constrain the access structure of the combinatorial holographic quantum secret sharing.

\section*{Acknowledgement}
We would like to thank Joydeep Naskar for discussions. N.\,B. is supported by the DOE Office of Science-ASCR, in particular the grant Novel Quantum Algorithms from Fast Classical Transforms and by Northeastern University. J.\,M. is supported by a graduate assistantship from Northeastern University. K.\,F. acknowledges support from Professor Ning Bao at Northeastern University.








\appendix

\section{Proofs} \label{app:proof_max_distance}

\subsection{Theorem \ref{prop:max_distance}} \label{app:proof_max_distance}

We provide lemma \ref{lem:erasure_equivalent_conditions} before we present the proof of theorem \ref{prop:max_distance}.

\begin{lemma}\label{lem:erasure_equivalent_conditions}
    For a choice of $n$, consider $\msR_{b}^{\min}\subseteq \mP[n]$ and $\mE_E \subseteq [n]$. The following statements are equivalent.
    \begin{enumerate}
        \item $|\msR_b^{\min}\setminus_E \mE_E |= 0$
        \item $\mE_E  \cap R \neq \emptyset$, $\forall R \in \msR_b^{\min}$ 
        \item $[n]\setminus \mE_E  \not\supseteq R$, $\forall R \in \msR_b^{\min}$
    \end{enumerate}
\end{lemma}
\begin{proof}
    $(1\Leftrightarrow 2)$The equivalence between the first and second statements can be seen from the definition of $\setminus_E$ in \eqref{eq:set_erasure}.

    We next prove the equivalence between the second and third statements. 
    
    $(2\Rightarrow 3)$ Assuming that $\mE_E  \cap R \neq \emptyset$ for $\forall R \in \msR_b^{\min}$, pick $A_i \in \mE_E \cap R$ for a given $R\in \msR_b^{\min}$. We then have $A_i \not\in [n]\setminus \mE_E $ and thus $ R \not \subseteq [n]\setminus \mE_E$. This holds for any $R\in\msR_b^{\min}$. Thus, the third statement is proved. 

    $(2\Leftarrow 3)$ We prove the contrapositive. If $\mE_E  \cap R = \emptyset$ for some $ R \in \msR_b^{\min}$, then $R$ is contained in the complement $[n]\setminus \mE_E$, i.e., $[n]\setminus \mE_E  \supseteq R$.

\end{proof}

We now begin the proof of theorem \ref{prop:max_distance}.
\begin{proof}
    For a choice of bulk subregion $b$ with $\msR_b\neq \emptyset$, define $u_b:= \underset{R\in \msR_b^{\min}}{\min} |R|$, the minimum size of a minimal authorized set. 
    
    We can choose $\mE_E =R$ for some $R\in \msR_b^{\min}$ with $|R|=u_b$, so that $|\msR_b^{\min}\setminus_E \mE_E |=0$ because, by lemma \ref{lem:pair-wise_intersection}, $R\in \msR_b^{\min}$ is pairwise intersecting with any other elements in $\msR_b^{\min}$. Thus, we have 
    \begin{equation} \label{eq:distance_upperbound1}
        d_b \leq u_b.
    \end{equation}
    
    We can alternatively choose $\mE_E $ such that 
    \begin{equation}\label{eq:condition_2_ver2}
        |[n]\setminus \mE_E | <u_b.
    \end{equation}
    From lemma \ref{lem:erasure_equivalent_conditions}, this choice of $\mE_E $ guarantees that $|\msR_
    b^{\min}\setminus \mE_E |=0$, because the complement $[n]\setminus_E \mE_E$ is too small to contain any element of $\msR_b^{\min}$. From \eqref{eq:condition_2_ver2}, we get
    \begin{equation}
        |\mE_E |> n-u_b, \text{ or } |\mE_E |\geq n-u_b+1.
    \end{equation}
    Since the choice of $\mE_E $ is not necessarily optimal, we have 
    \begin{equation}\label{eq:distance_upperbound2}
        d_b  \leq n-u_b +1.
    \end{equation}
    We choose the minimum upper bound of $d_b$ between \eqref{eq:distance_upperbound1} and \eqref{eq:distance_upperbound2}, i.e.,
    \begin{equation}
        d_b \leq \min(u_b,n-u_b+1).
    \end{equation}
    We then evaluate the maximum distance scanning over $b$. That is,
    \begin{equation}
       d_b \leq d_{\max} = \underset{b}{\max} \min (u_b,n-u_b+1).
    \end{equation}

    If $n$ is even, let $u_b=\frac{n}{2}$. Then, 
    \begin{equation}
        \underset{b}{\max} \min \lp \frac{n}{2},\frac{n}{2}+1 \rp =\frac{n}{2}.
    \end{equation}
    If $n$ is odd, let $u_b=\frac{n+1}{2}$. Then, 
    \begin{equation}
        \underset{b}{\max} \min \lp \frac{n+1}{2},\frac{n+1}{2}+1 \rp =\frac{n+1}{2}.
    \end{equation}
    Therefore, we have
    \begin{equation}
        d_{max} = \Big\lceil \frac{n}{2}\Big\rceil.
    \end{equation}

    We should question whether $d_{max}$ is physically realizable. The answer is positive, at least in pure AdS$_3$, if the entanglement wedges of all $\Big\lceil \frac{n+1}{2} \Big\rceil$ boundary subregions are connected. If so, we can find $b$ such that $|R|=\Big\lceil \frac{n+1}{2} \Big\rceil$ for $\forall R\in \msR_b^{\min}$ and 
    \begin{equation}
        |\msR_b^{\min}| = {}_n C_{\big\lceil \frac{n+1}{2} \big\rceil}.    
    \end{equation}

    There are ${}_{n-1}C_{\big\lceil \frac{n+1}{2} \big\rceil-1}$ elements in $\msR_b^{\min}$, which contain boundary subregion $A_{i_1}\in [n]$. They can be removed from $\msR_b^{\min}$ if $A_{i_1}\in \mE_E $. Similarly, there are ${}_{n-2}C_{\big\lceil \frac{n+1}{2} \big\rceil-1}$ elements in $\msR_b^{\min}$, which contain boundary subregions $A_{i_1},A_{i_2}\in [n]$. If $A_{i_2}\in \mE_E $ as well, they can be removed from $\msR_b^{\min}$. 

    We can add
    \begin{equation}
        (n-1) - \lp \left\lceil \frac{n+1}{2} \right\rceil-1\rp + 1 = \left\lceil \frac{n}{2} \right\rceil
    \end{equation}
    boundary subregions to $\mE_E $. The total number of elements in $\msR_b^{\min}$ that are removed after the size $|\mE_E |$ becomes $\Big\lceil \frac{n}{2} \Big\rceil$ is
    \begin{equation}
        {}_n C_{\big\lceil \frac{n+1}{2} \big\rceil}  = {}_{n-1} C_{\big\lceil \frac{n+1}{2} \big\rceil-1} + \cdots + {}_{\big\lceil \frac{n+1}{2} \big\rceil-1} C_{\big\lceil \frac{n+1}{2} \big\rceil-1}.
    \end{equation}
    Therefore, we can achieve $|\mE_E | = d_{max} = \big\lceil \frac{n}{2} \big\rceil$ such that $|\msR_b^{\min}\setminus \mE_E |=0$ with the set $\msR_b^{\min}$.

\end{proof}

\subsection{Theorem \ref{thm:secret_distance}}\label{app:proof_secret_distance}

We prove theorem \ref{thm:secret_distance} below. 

\begin{proof}
    Consider $\mE_E \in [n]$ such that $|\mE_E| =s_b$. By the definition of $s_b$, we have 
    \begin{equation}
        b\not\subseteq EW(\mE_E).
    \end{equation}

    $(\Rightarrow)$ If $b\not\subseteq EW(\mE_E)$ and $b\subseteq EW(R = [n]\setminus \mE_E)$, we have
    \begin{equation}
        n-s_b =  r_b.
    \end{equation}
    From \eqref{eq:reconstruction_distance}, we get
    \begin{equation}
        s_b = d_b-1.
    \end{equation}
    However, if $b\not\subseteq EW(\mE_E)$ and $b\not\subseteq EW(R=[n]\setminus \mE_E)$, we have 
    \begin{equation}
        r_b \geq \max(s_b+1,n-s_b+1).
    \end{equation}
    If $n>2s_b$, 
    \begin{equation}
        r_b \geq n-s_b+1.
    \end{equation}
    Using \eqref{eq:reconstruction_distance}, we get 
    \begin{equation}
        s_b\geq d_b.
    \end{equation}
    On the contrary, if $n\leq 2s_b$,
    \begin{equation}
        r_b \geq s_b+1.
    \end{equation}
    From \eqref{eq:reconstruction_distance}, 
    \begin{equation}
        n-d_b \geq s_b.
    \end{equation}
    Since $n\leq 2s_b$, we get
    \begin{equation}
        2s_b -d_b \geq n-d_b \geq s_b.
    \end{equation}
    Thus, we again obtain $s_b \geq d_b$.



    $(\Leftarrow)$ If $s_b =d_b-1$ for $\mE_E$ such that $|\mE_E|=s_b$, we have 
    \begin{equation}
        s_b = n-r_b
    \end{equation}
    from \eqref{eq:reconstruction_distance}. This implies that $b\subseteq EW(R \setminus \mE_E)$. Similarly, if $s_b \geq d_b$, then $r_b \geq n-s_b +1$. Hence, $b\not\subseteq EW(R \setminus \mE_E)$. This proves the sufficiency of \eqref{eq:secret_distance} and \eqref{eq:secret_distance_superadditive}.
    
\end{proof}

\section{More examples}\label{app:examples}

In this section, we present examples of pure-state and mixed-state CHQSS schemes in the symmetric cases for distinct holographic phases. In particular, in section \ref{app:examples_pure_state}, we describe the corresponding family of pure-state CHQSS schemes for $n=3,4,5$. In the following two subsections, we explore families of mixed-state CHQSS schemes in distinct holographic phases for $n=3,5$.


\subsection{Pure-state CHQSS schemes: $n=3,4,5$}\label{app:examples_pure_state}

We have figure \ref{fig:examples_pure_n-3_graphs}, \ref{fig:examples_pure_n-4_graphs}, and \ref{fig:examples_pure_n-5_graphs} of a constant time slice of AdS$_3$ partitioned by RT surfaces anchored on boundary subregions in $\mP[n]$ for $n=3,4,5$.

As we have proved in theorem \ref{thm:secret_distance} and corollary \ref{cor:pure-mixed}, we have $d_b-s_b=1$ for any $b$ and $n$, see figures \ref{fig:examples_pure_n-3_plots}-(b), \ref{fig:examples_pure_n-3_plots}-(b), and \ref{fig:examples_pure_n-3_plots}-(b). The gap $r_b-s_b$ for any $b$ increases monotonically as $b$ moves away from the center of the region graph because $r_b$ increases with increasing inhomogeneity of $\msR_b$, while $s_b$ decreases.

\begin{figure}[ht]
    \centering
    \includegraphics[width=0.8\linewidth]{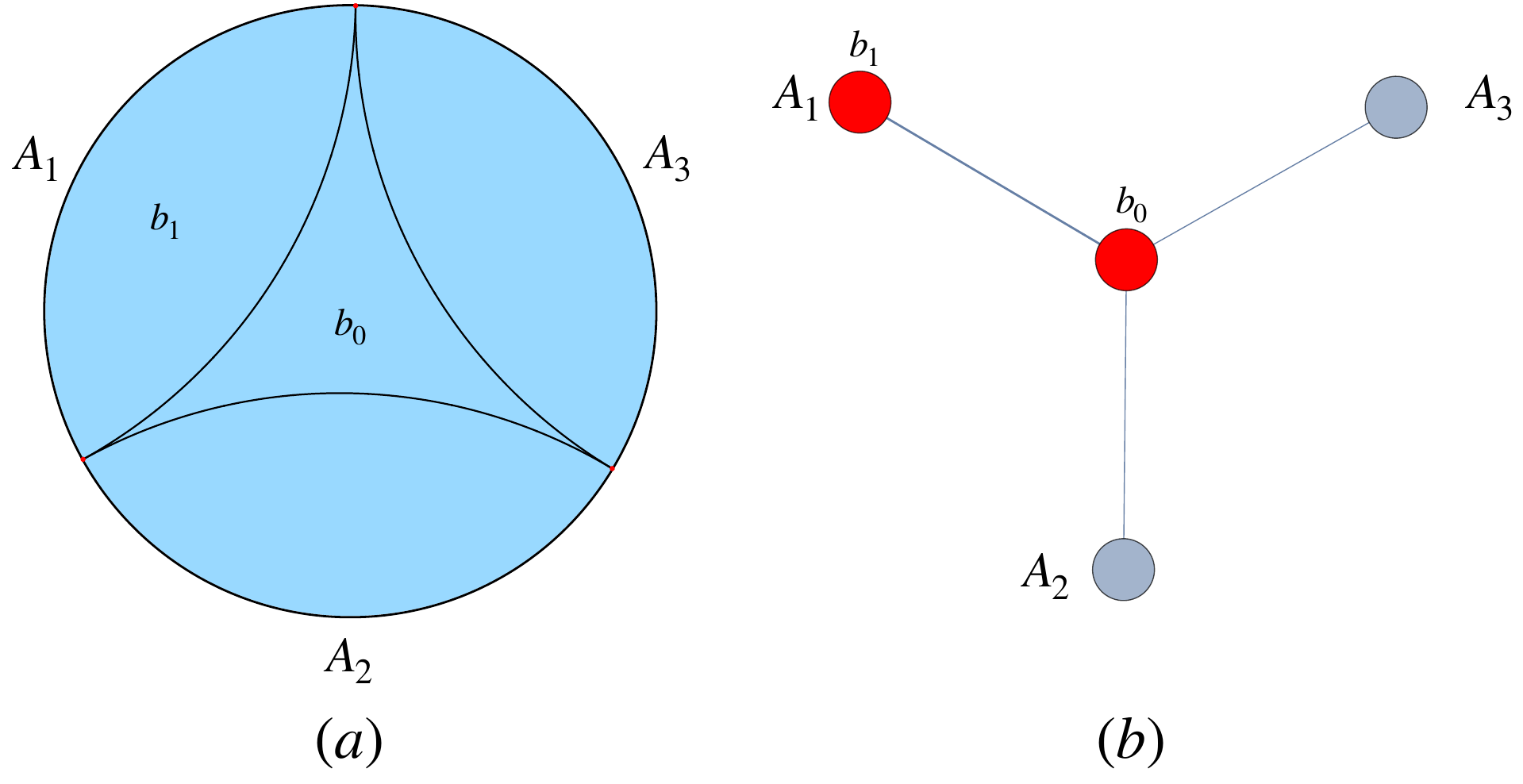}
    \caption{\small{(a) The constant time slice partitioned by RT surfaces of boundary subregions in $\mP[n=3]$ at $\phi =2\pi/3$. (b) The region graph for $n=3$. The vertices $b_0$ and $b_1$ in the path from the center of the region graph to the boundary vertex $A_1$ are colored red.}}
    \label{fig:examples_pure_n-3_graphs}
\end{figure}

\begin{figure}[ht]
    \centering
    \includegraphics[width=0.9\linewidth]{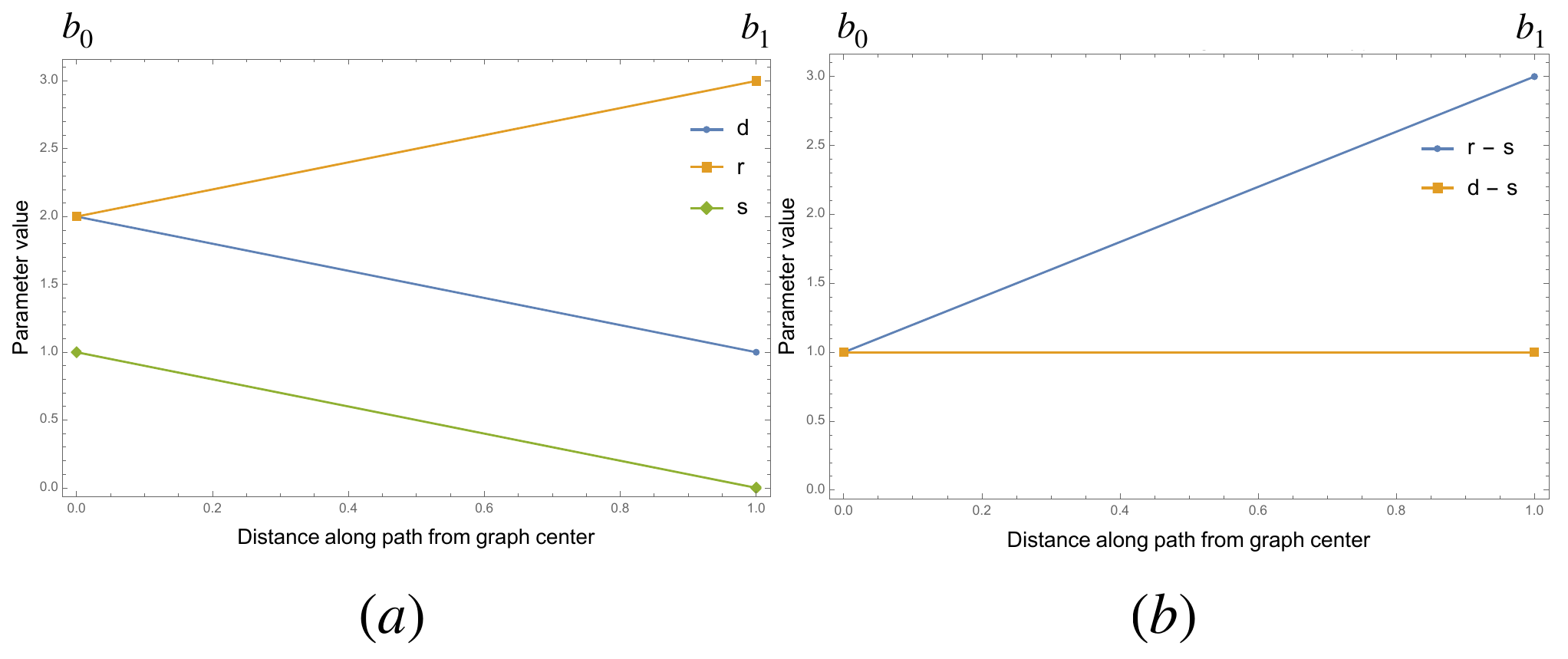}
    \caption{\small{The plots of the parameters $d_b$, $r_b$, $s_b$, and the differences $r_b-s_b$ and $d_b-s_b$ as a function of the distance along the path from the center of the graph: (a) $d_b$, $r_b$, and $s_b$. (b) $r_b-s_b$, and $d_b-s_b$.}}
    \label{fig:examples_pure_n-3_plots}
\end{figure}

\begin{figure}[ht]
    \centering
    \includegraphics[width=0.8\linewidth]{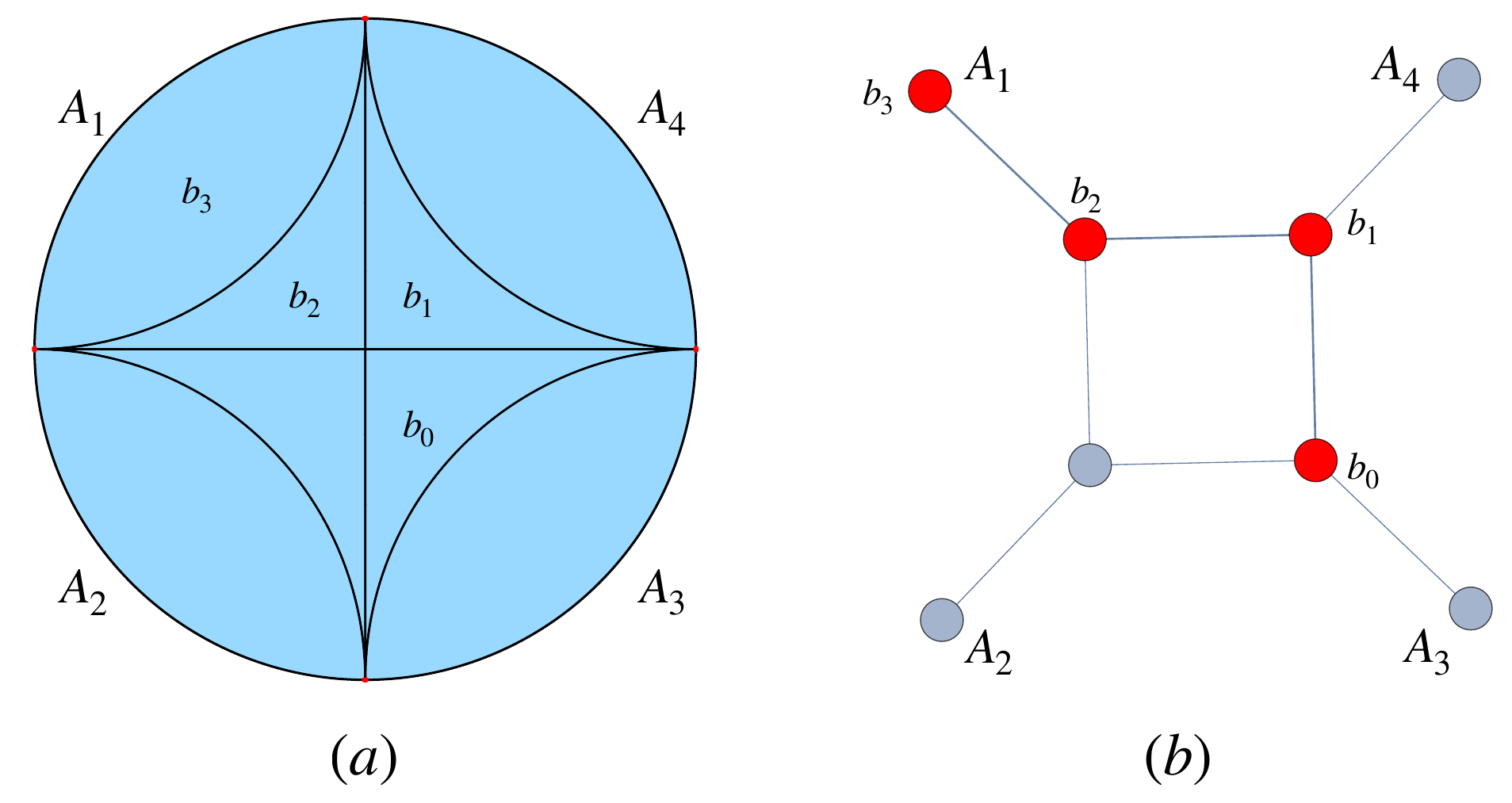}
    \caption{\small{(a) The constant time slice partitioned by RT surfaces of boundary subregions in $\mP[n=4]$ at $\phi =2\pi/4$. (b) The region graph for $n=4$. The vertices $b_0$, $b_1$, $b_2$ and $b_3$ in the path from the center of the region graph to the boundary vertex $A_1$ are colored red.}}
    \label{fig:examples_pure_n-4_graphs}
\end{figure}

\begin{figure}[ht]
    \centering
    \includegraphics[width=0.9\linewidth]{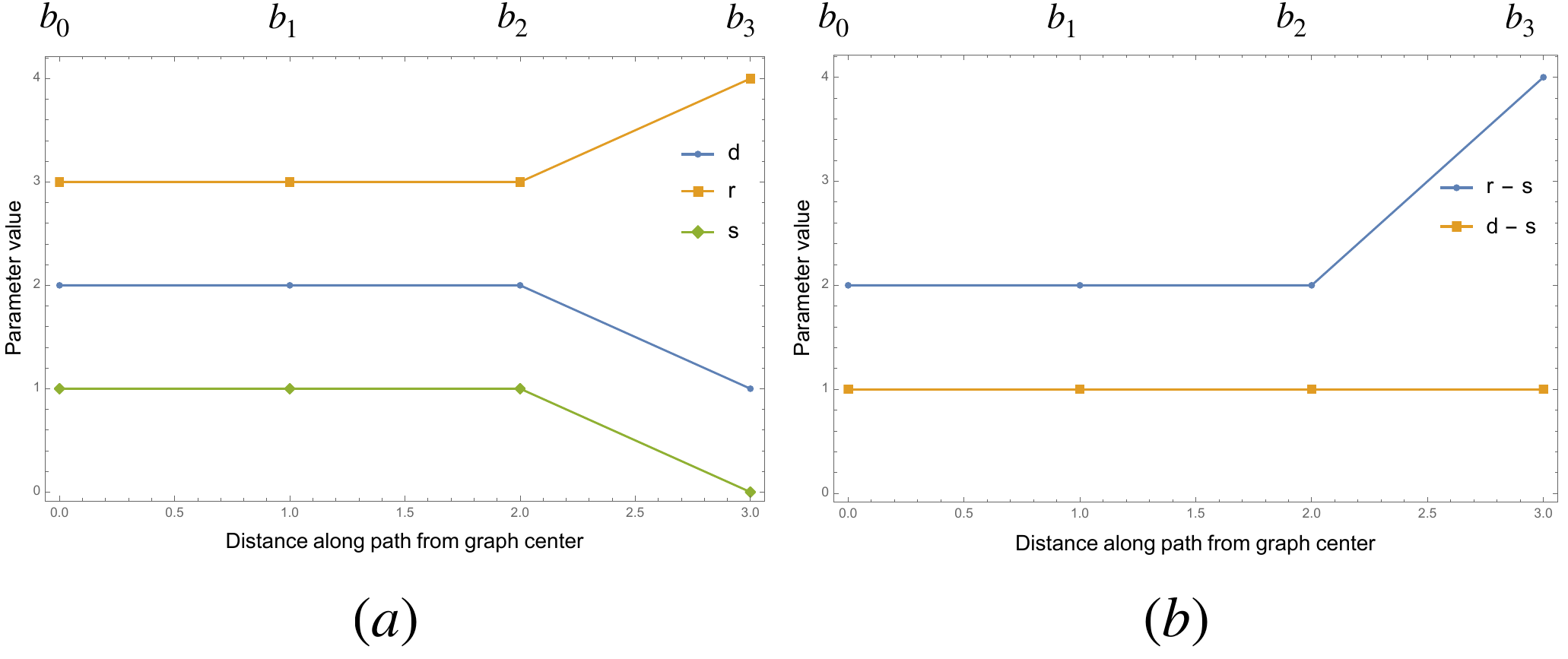}
    \caption{\small{$n=4$, pure-state CHQSS. The plots of the parameters $d_b$, $r_b$, $s_b$, and the differences $r_b-s_b$ and $d_b-s_b$ as a function of the distance along the path from the center of the graph: (a) $d_b$, $r_b$, and $s_b$. (b) $r_b-s_b$, and $d_b-s_b$.}}
    \label{fig:examples_pure_n-4_plots}
\end{figure}

\begin{figure}[ht]
    \centering
    \includegraphics[width=0.8\linewidth]{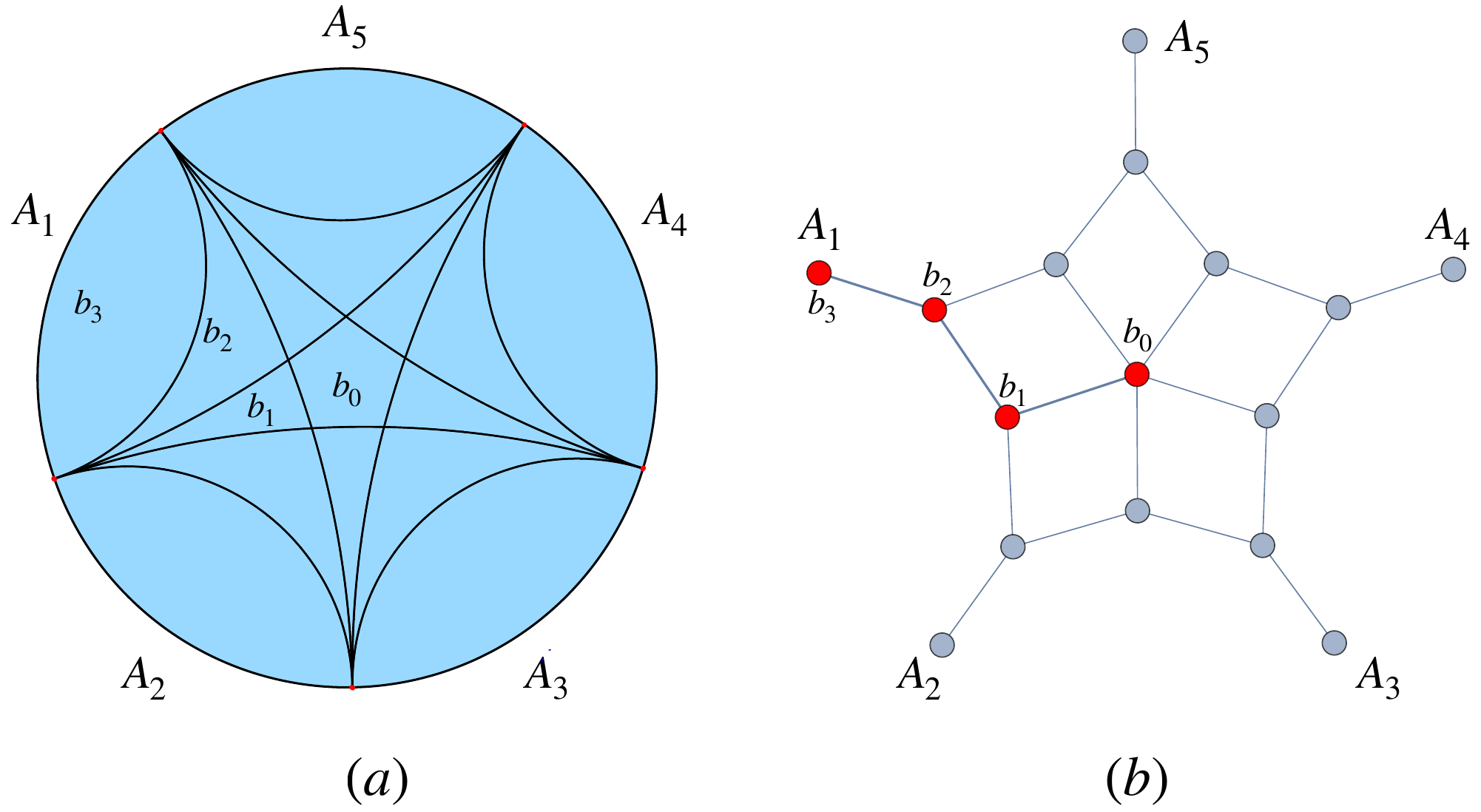}
    \caption{\small{(a) The constant time slice partitioned by RT surfaces of boundary subregions in $\mP[n=5]$ at $\phi =2\pi/5$. (b) The region graph for $n=5$. The vertices $b_0$, $b_1$, $b_2$ and $b_3$ in the path from the center of the region graph to the boundary vertex $A_1$ are colored red.}}
    \label{fig:examples_pure_n-5_graphs}
\end{figure}

\begin{figure}[ht]
    \centering
    \includegraphics[width=0.9\linewidth]{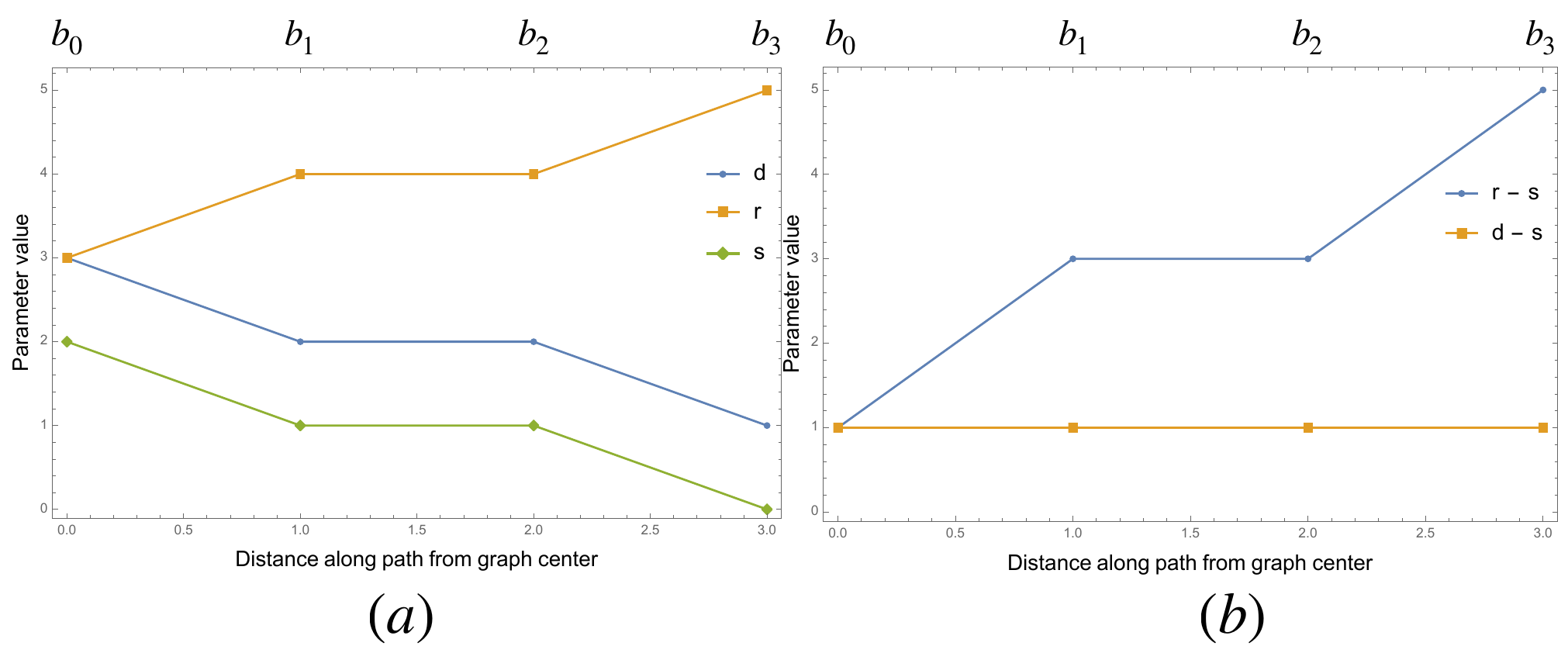}
    \caption{\small{$n=5$, pure-state CHQSS. The plots of the parameters $d_b$, $r_b$, $s_b$, and the differences $r_b-s_b$ and $d_b-s_b$ as a function of the distance along the path from the center of the graph: (a) $d_b$, $r_b$, and $s_b$. (b) $r_b-s_b$, and $d_b-s_b$.}}
    \label{fig:examples_pure_n-5_plots}
\end{figure}

\subsection{Mixed-state CHQSS  schemes: $n=3$}

For $n=3$, there are only two connected phases as described in figure \ref{fig:transition_points_n-3-4}-(a). We have already studied the family of mixed-state CHQSS schemes in the phase where $\phi \geq \phi_1$ in section \ref{sec:CHQSS}. Here, we complete the analysis of $n=3$ CHQSS by presenting the family of mixed-state schemes in the phase where $\phi_0 \leq \phi \leq \phi_1$, see figure \ref{fig:examples_mixed_n-3_graphs} and \ref{fig:examples_mixed_n-3_plots}. One can see that $(\!(r_{b_0},s_{b_0},n)\!)$-CHQSS has $d_b-s_b=-1$ and hence it is superadditive. 

Including the pure-state case, we have three types of $(\!(r_{b_0},s_{b_0},n)\!)$-CHQSS schemes, the first is the pure-state CHQSS in figure \ref{fig:examples_pure_n-3_graphs}, the second is the mixed-state CHQSS in figure \ref{fig:graphs_n-3_d_s_difference_pure_mixed}-(b) where all the boundary subregions in $\mP[3]$ are connected, and the third is the mixed-state CHQSS in figure \ref{fig:examples_mixed_n-3_graphs} where only the tripartite entanglement wedge is connected. In figure \ref{fig:examples_pure_n-3_plots} and \ref{fig:plots_n-3_graph-center_d-r-s_differences}, we have 
\begin{equation}
    r_{b_0} = 2,\;s_{b_0} = 1,\;d_{b_0} = 2    
\end{equation}
for the first two CHQSS schemes, and, in figure \ref{fig:examples_mixed_n-3_plots},
\begin{equation}
    r_{b_0} = 3,\;s_{b_0} = 2,\;d_{b_0} = 1    
\end{equation}
for the last CHQSS. The first two $(\!(r_{b_0},s_{b_0},n)\!)$-CHQSS schemes are additive since $d_{b_0}-s_{b_0} = 1$ regardless of pure-state or mixed-state scheme, whereas the third one is superadditive since $d_{b_0}-s_{b_0} = -1$. In addition, the bulk logical information in $b_0$ in the third case is more fragile than the first two because its $d_{b_0}$ is smaller than theirs. The $d_{b_0}$ of the first two reaches the maximum, i.e., $d_{max}=\lceil 3/2 \rceil = 2$.

\begin{figure}[ht]
    \centering
    \includegraphics[width=0.8\linewidth]{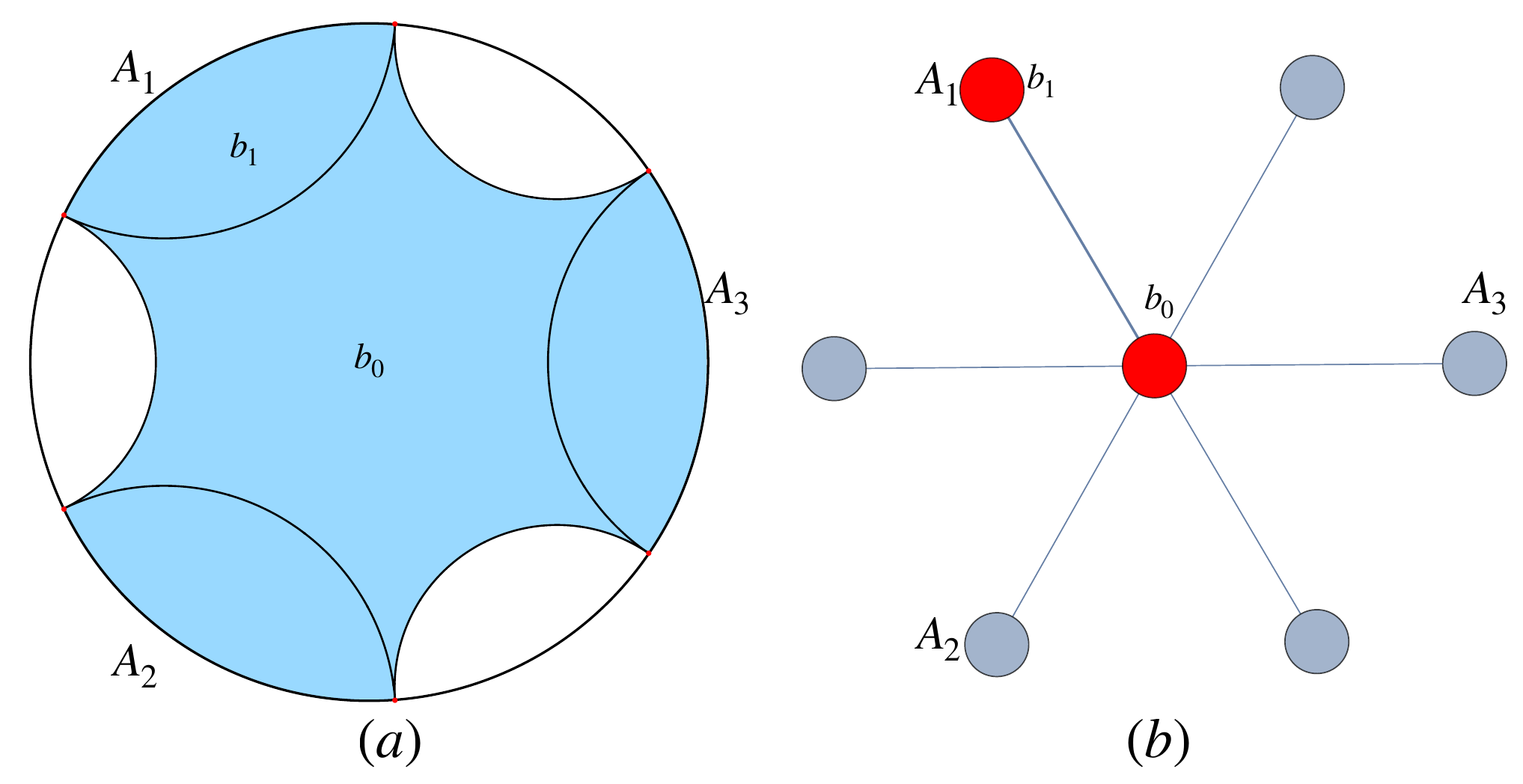}
    \caption{\small{(a) The constant time slice partitioned by RT surfaces of boundary subregions in $\mP[n=3]$ in the phase $\phi_0\leq \phi \leq \phi_1$. (b) The region graph for $n=3$. The vertices $b_0$ and $b_1$ in the path from the center of the region graph to the boundary vertex $A_1$ are colored red.}}
    \label{fig:examples_mixed_n-3_graphs}
\end{figure}

\begin{figure}[ht]
    \centering
    \includegraphics[width=0.9\linewidth]{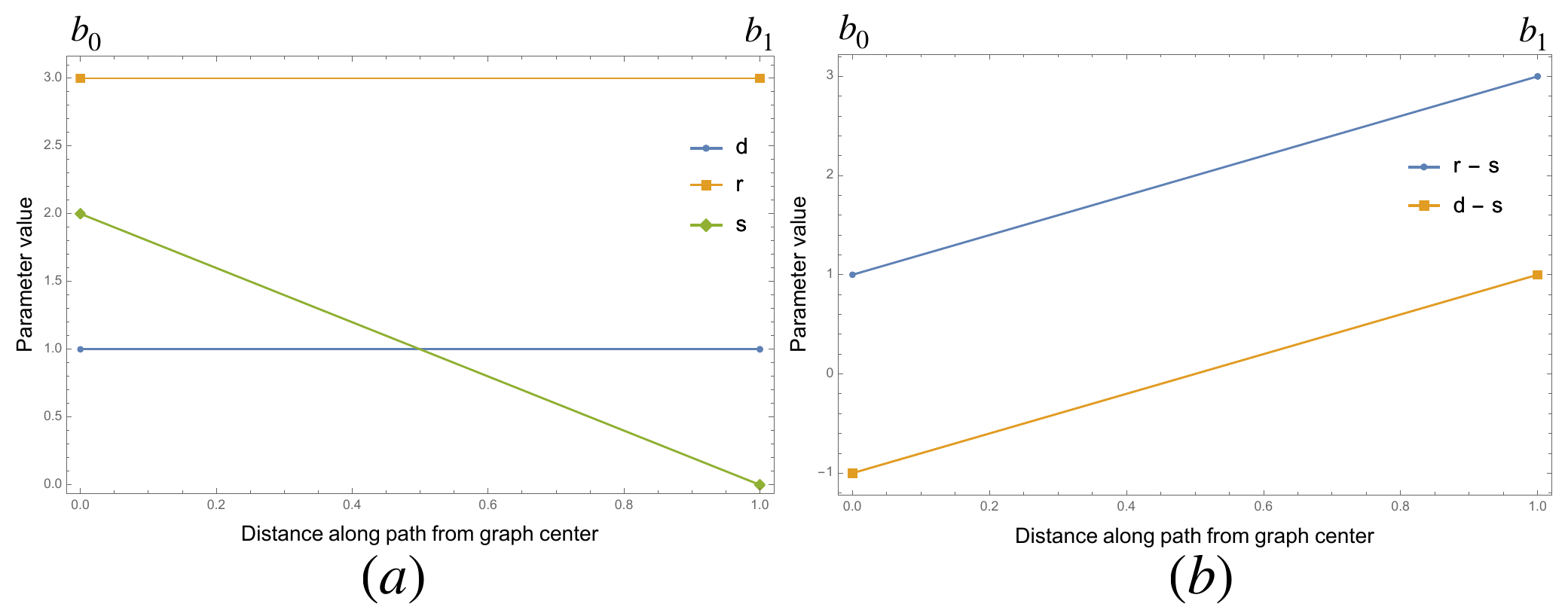}
    \caption{\small{$n=3$, mixed-state CHQSS. The plots of the parameters $d_b$, $r_b$, $s_b$, and the differences $r_b-s_b$ and $d_b-s_b$ as a function of the distance along the path from the center of the graph: (a) $d_b$, $r_b$, and $s_b$. (b) $r_b-s_b$, and $d_b-s_b$.}}
    \label{fig:examples_mixed_n-3_plots}
\end{figure}

\subsection{Mixed-state CHQSS schemes: $n=5$}

In this last subsection, we present the family of mixed-state CHQSS schemes for $n=5$, see figures \ref{fig:examples_mixed_n-5_70_graphs}, \ref{fig:examples_mixed_n-5_75_graphs}, \ref{fig:examples_mixed_n-5_82_graphs},
\ref{fig:examples_mixed_n-5_95_graphs}, and \ref{fig:examples_mixed_n-5_11_graphs} for the geometries and the corresponding region graphs, and figures \ref{fig:examples_mixed_n-5_70_plots}, \ref{fig:examples_mixed_n-5_75_plots}, \ref{fig:examples_mixed_n-5_82_plots},
\ref{fig:examples_mixed_n-5_95_plots}, and \ref{fig:examples_mixed_n-5_11_plots} for the behaviors of the parameters $r_b$, $s_b$, and $d_b$. There are five non-trivial transition points $\phi_0$, $\phi_1$, $\phi_2$, $\phi_3$, and $\phi_4$, and, thus, five distinct connected phases.

When $\phi_0 \leq \phi \leq \phi_1$, only five-partite entanglement wedge of $(1,1,1,1,1)$ is connected. In the phase of $\phi_1 \leq \phi \leq \phi_2$, the four-partite entanglement wedges of the boundary subregion $(1,1,1,2)$ are connected. The phase $\phi_2 \leq \phi \leq \phi_3$ has the tripartite entanglement wedges of the boundary subregions labeled $(1,1,3)$ connected. Within the phase $\phi_3 \leq \phi \leq \phi_4$, the bipartite entanglement wedges of $(1,4)$ begin to connect. After passing the last transition point $ \phi \geq \phi_4 $, another set of tripartite entanglement wedges of $(1,2,2)$ appears as explained in section \ref{sec:transition_points_MEW}.

In the phase $\phi_0 \leq \phi < \phi_3$, only the CHQSS schemes on the boundary vertex are additive. The other CHQSS schemes on the bulk vertices are superadditive because they are not covered by enough layers of multipartite entanglement wedges. In the phase after $\phi_3$, some superadditive CHQSS schemes on certain bulk vertices emerge. For example, the CHQSS of $b_4$ is shown in figure \ref{fig:examples_mixed_n-5_95_plots} and \ref{fig:examples_mixed_n-5_11_plots}.

The gap $d_{b_0}-s_{b_0}$ of the CHQSS at the center of the region graph in each phase increases monotonically as the phase shifts from $\phi_0$ to above $\phi_4$. That is, $d_{b_0}-s_{b_0} = -3 \text{ in }\phi_0 \leq \phi < \phi_1$, $-1 \text{ in }\phi_1 \leq \phi < \phi_2$, $0 \text{ in }\phi_2 \leq \phi < \phi_3$, $0 \text{ in }\phi_3 \leq \phi < \phi_4$, and $1 \text{ in }\phi_4 \leq \phi \leq 2\pi/5$. Thus, the CHQSS in $\phi_4 \leq \phi \leq 2\pi/5$ is additive. Its distance reaches the maximum $d_{b_0}= \lceil 5/2 \rceil=3$.

\begin{figure}[ht]
    \centering
    \includegraphics[width=0.8\linewidth]{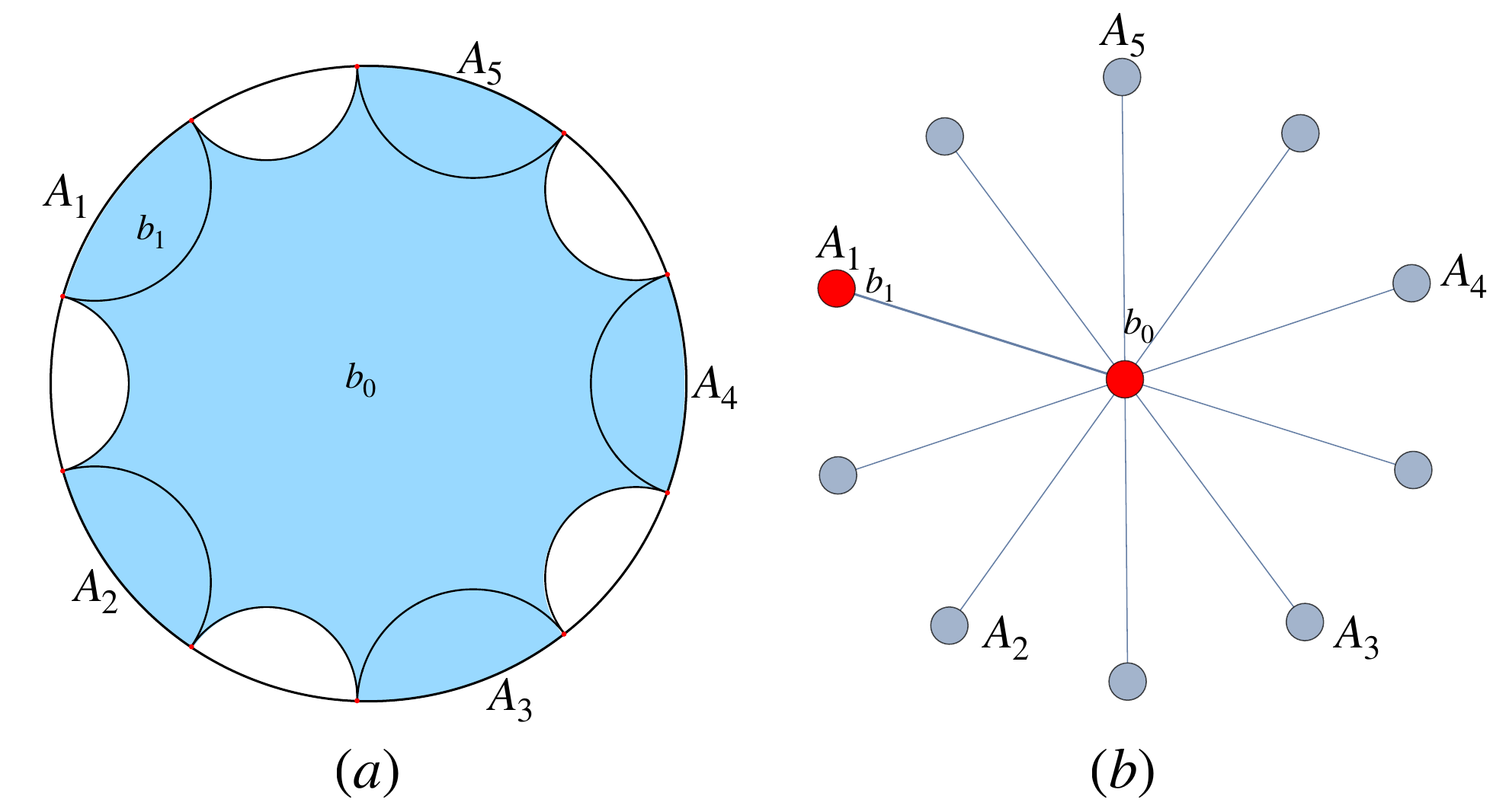}
    \caption{\small{(a) The constant time slice partitioned by RT surfaces of boundary subregions in $\mP[n=5]$ in the phase $\phi_0\leq \phi < \phi_1$. (b) The region graph for $n=5$. The vertices $b_0$ and $b_1$ in the path from the center of the region graph to the boundary vertex $A_1$ are colored red.}}
    \label{fig:examples_mixed_n-5_70_graphs}
\end{figure}

\begin{figure}[ht]
    \centering
    \includegraphics[width=0.9\linewidth]{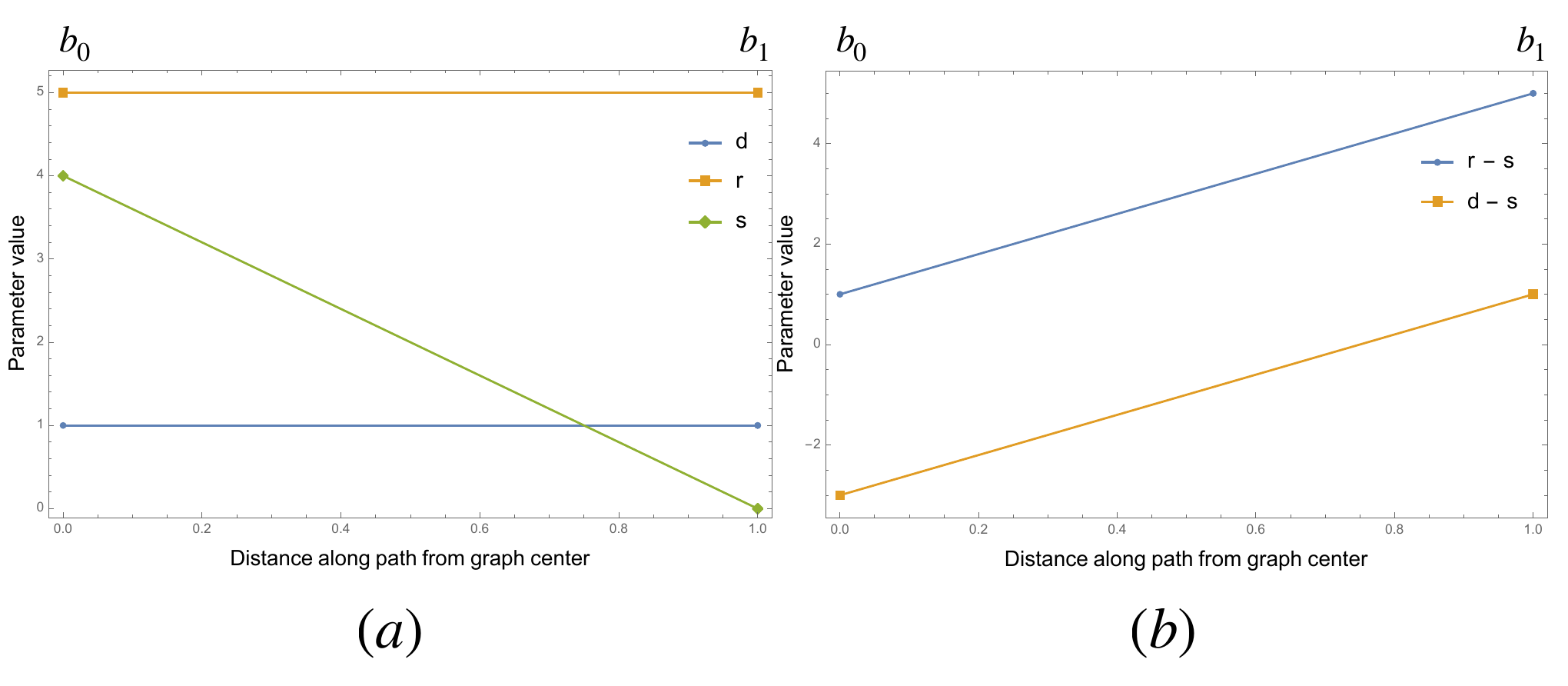}
    \caption{\small{$n=5$, mixed-state CHQSS. The plots of the parameters $d_b$, $r_b$, $s_b$, and the differences $r_b-s_b$ and $d_b-s_b$ as a function of the distance along the path from the center of the graph: (a) $d_b$, $r_b$, and $s_b$. (b) $r_b-s_b$, and $d_b-s_b$.}}
    \label{fig:examples_mixed_n-5_70_plots}
\end{figure}

\begin{figure}[ht]
    \centering
    \includegraphics[width=0.8\linewidth]{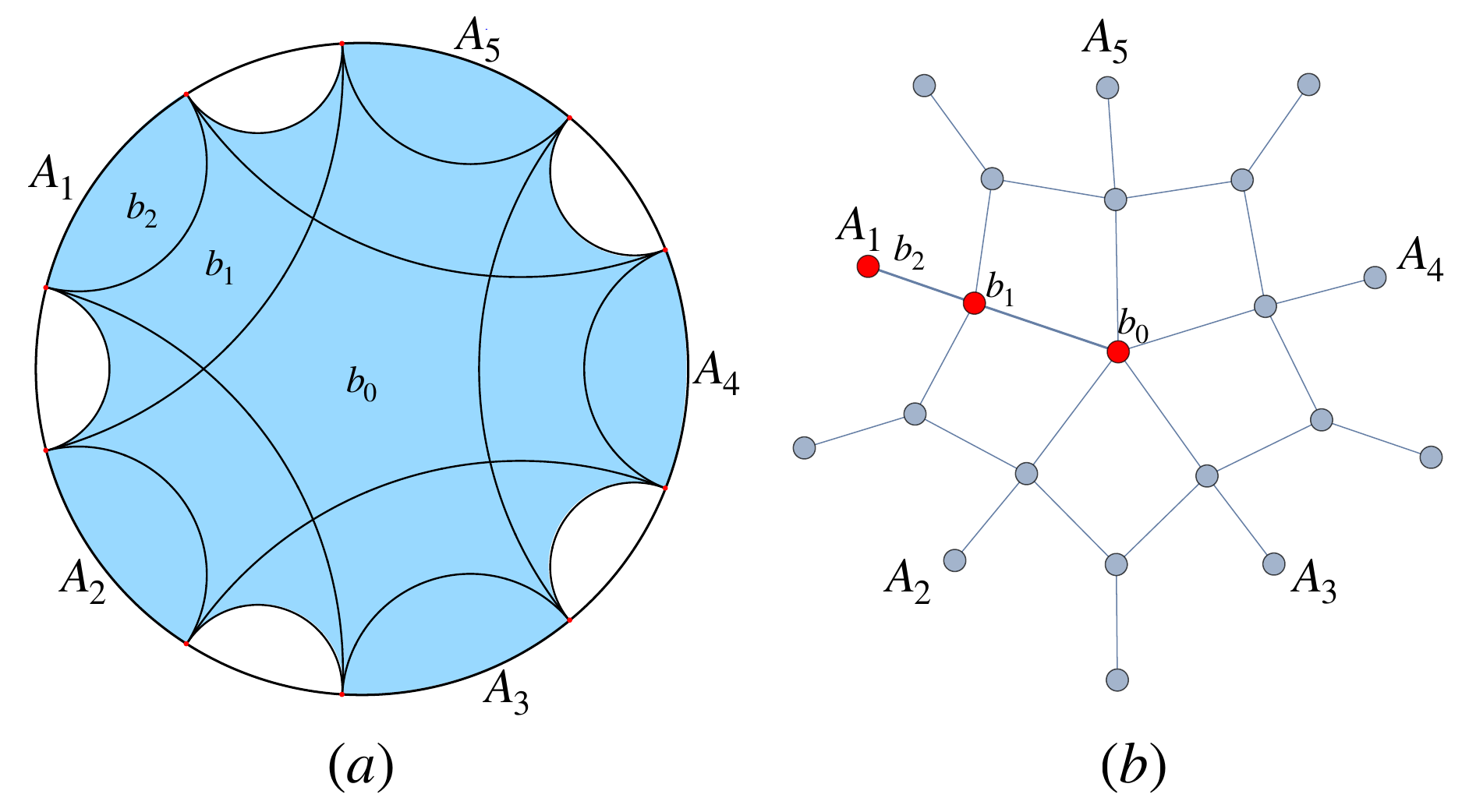}
    \caption{\small{(a) The constant time slice partitioned by RT surfaces of boundary subregions in $\mP[n=5]$ in the phase $\phi_1\leq \phi < \phi_2$. (b) The region graph for $n=5$. The vertices $b_0$, $b_1$, and $b_2$ in the path from the center of the region graph to the boundary vertex $A_1$ are colored red.}}
    \label{fig:examples_mixed_n-5_75_graphs}
\end{figure}

\begin{figure}[ht]
    \centering
    \includegraphics[width=0.9\linewidth]{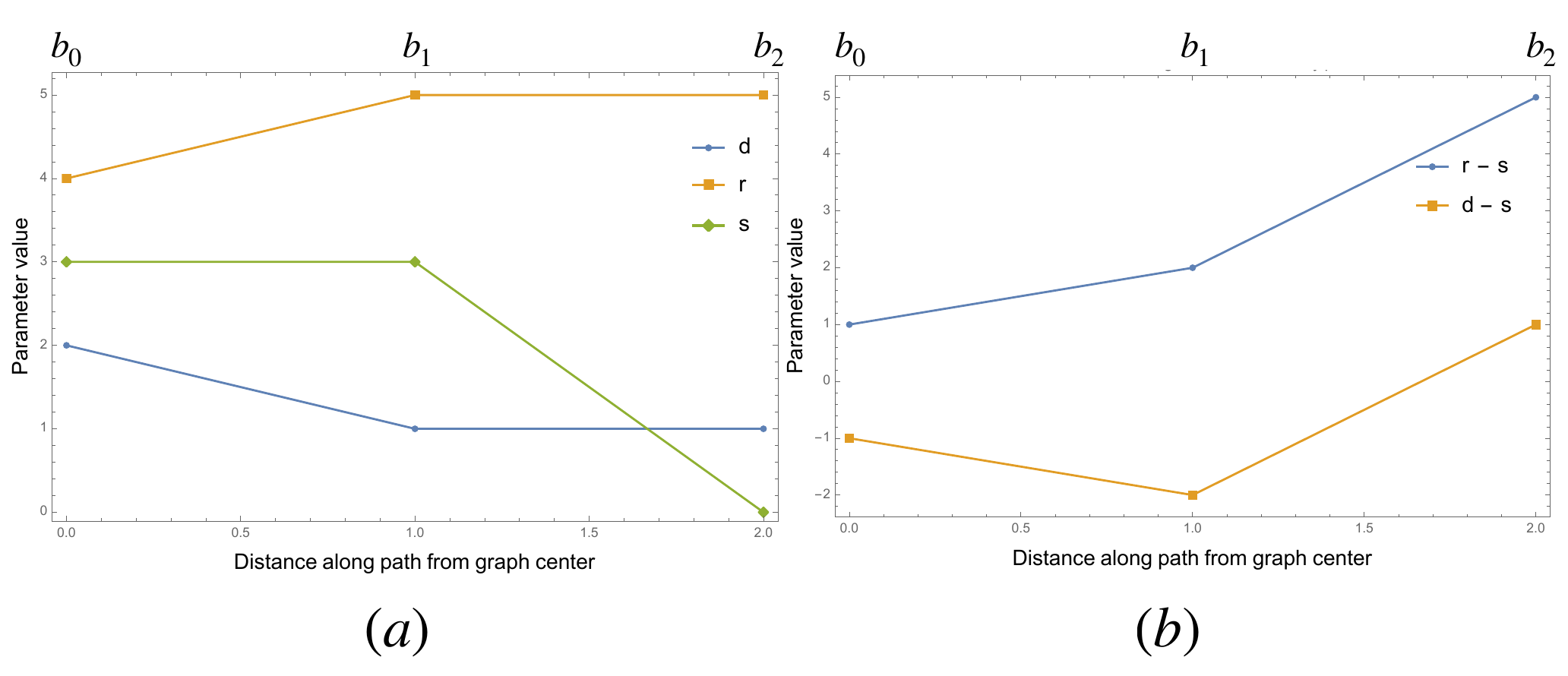}
    \caption{\small{$n=5$, mixed-state CHQSS. The plots of the parameters $d_b$, $r_b$, $s_b$, and the differences $r_b-s_b$ and $d_b-s_b$ as a function of the distance along the path from the center of the graph: (a) $d_b$, $r_b$, and $s_b$. (b) $r_b-s_b$, and $d_b-s_b$.}}
    \label{fig:examples_mixed_n-5_75_plots}
\end{figure}

\begin{figure}[ht]
    \centering
    \includegraphics[width=0.8\linewidth]{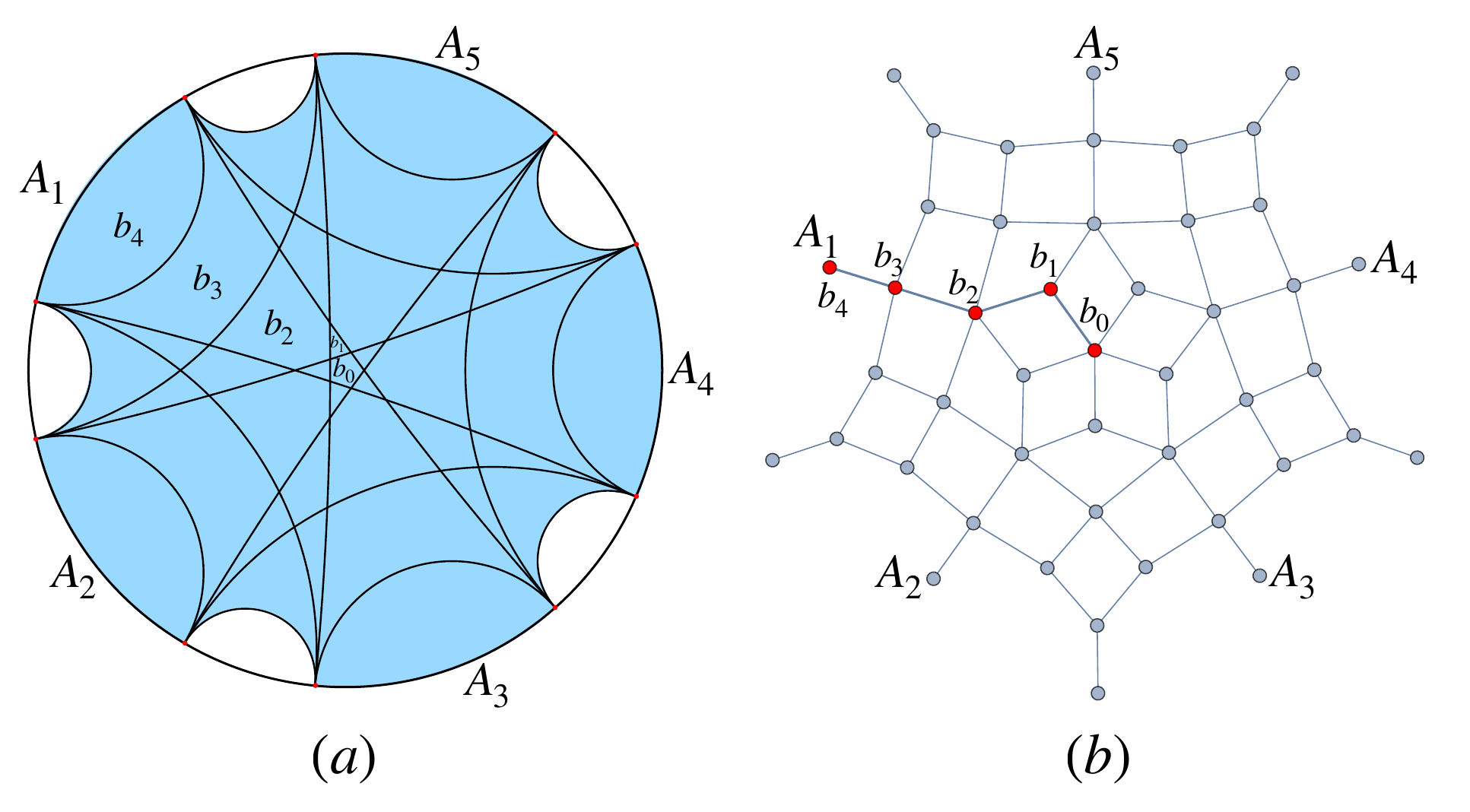}
    \caption{\small{(a) The constant time slice partitioned by RT surfaces of boundary subregions in $\mP[n=5]$ in the phase $\phi_2\leq \phi < \phi_3$. (b) The region graph for $n=5$. The vertices $b_0$, $b_1$, $b_2$, $b_3$ and $b_4$ in the path from the center of the region graph to the boundary vertex $A_1$ are colored red.}}
    \label{fig:examples_mixed_n-5_82_graphs}
\end{figure}

\begin{figure}[ht]
    \centering
    \includegraphics[width=0.9\linewidth]{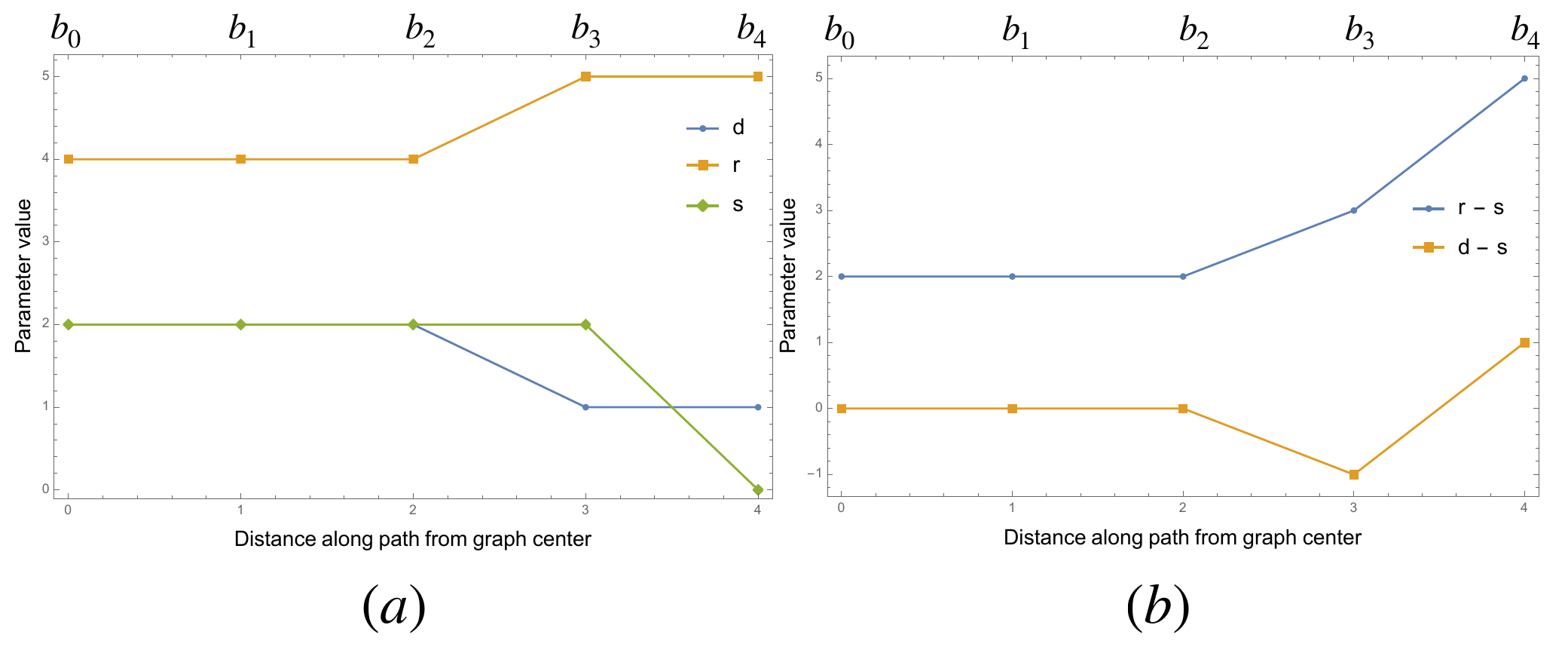}
    \caption{\small{$n=5$, mixed-state CHQSS. The plots of the parameters $d_b$, $r_b$, $s_b$, and the differences $r_b-s_b$ and $d_b-s_b$ as a function of the distance along the path from the center of the graph: (a) $d_b$, $r_b$, and $s_b$. (b) $r_b-s_b$, and $d_b-s_b$.}}
    \label{fig:examples_mixed_n-5_82_plots}
\end{figure}

\begin{figure}[ht]
    \centering
    \includegraphics[width=0.8\linewidth]{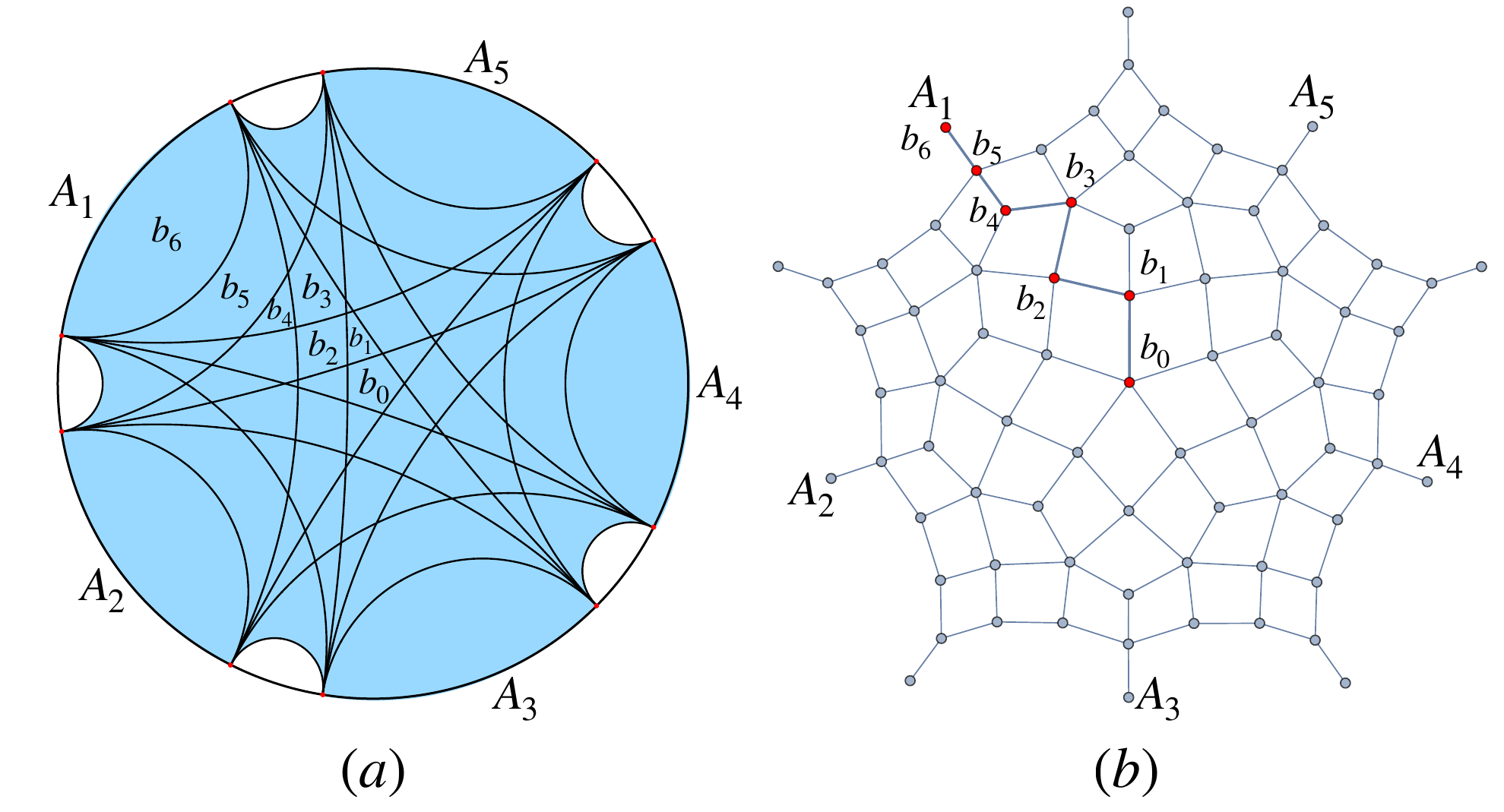}
    \caption{\small{(a) The constant time slice partitioned by RT surfaces of boundary subregions in $\mP[n=5]$ in the phase $\phi_3\leq \phi < \phi_4$. (b) The region graph for $n=5$. The vertices $b_0$, $b_1$, $b_2$, $b_3$, $b_4$, $b_5$ and $b_6$ in the path from the center of the region graph to the boundary vertex $A_1$ are colored red.}}
    \label{fig:examples_mixed_n-5_95_graphs}
\end{figure}

\begin{figure}[ht]
    \centering
    \includegraphics[width=0.9\linewidth]{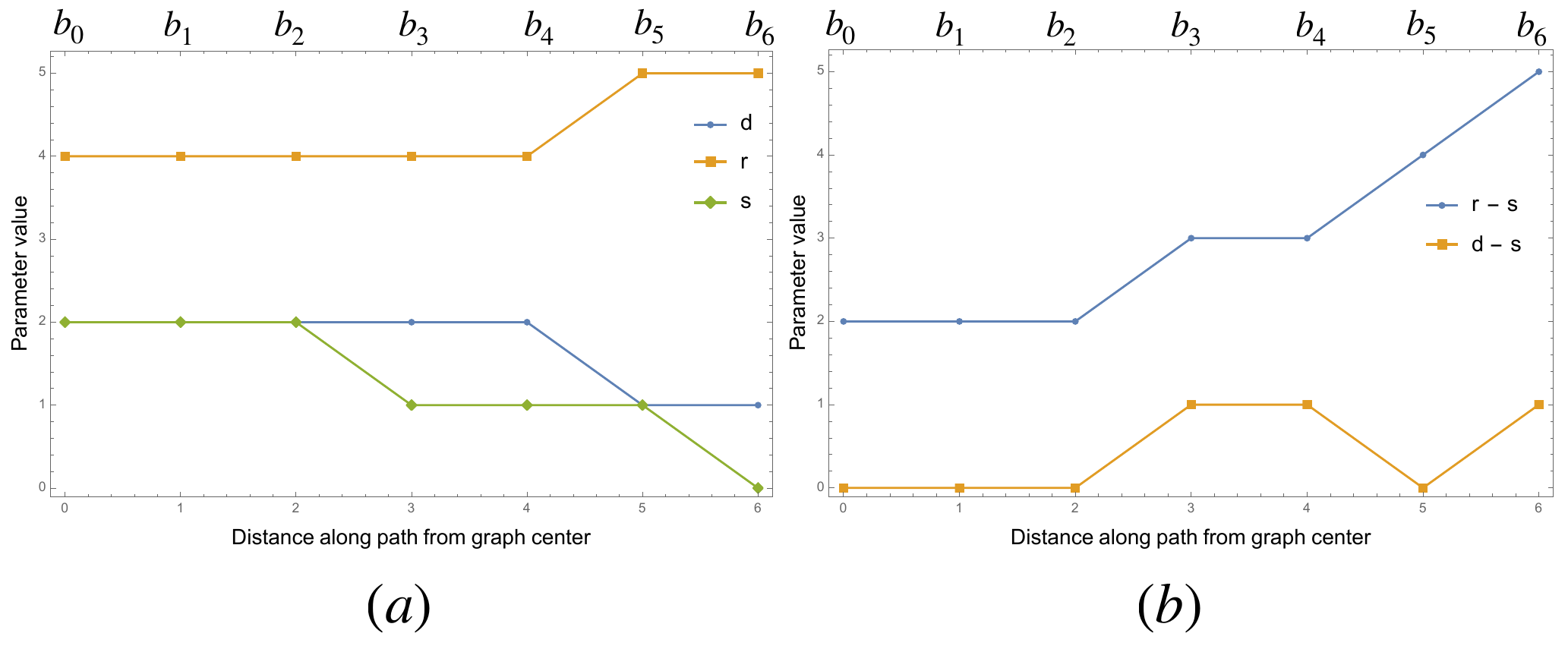}
    \caption{\small{$n=5$, mixed-state CHQSS. The plots of the parameters $d_b$, $r_b$, $s_b$, and the differences $r_b-s_b$ and $d_b-s_b$ as a function of the distance along the path from the center of the graph: (a) $d_b$, $r_b$, and $s_b$. (b) $r_b-s_b$, and $d_b-s_b$.}}
    \label{fig:examples_mixed_n-5_95_plots}
\end{figure}

\begin{figure}[ht]
    \centering
    \includegraphics[width=0.8\linewidth]{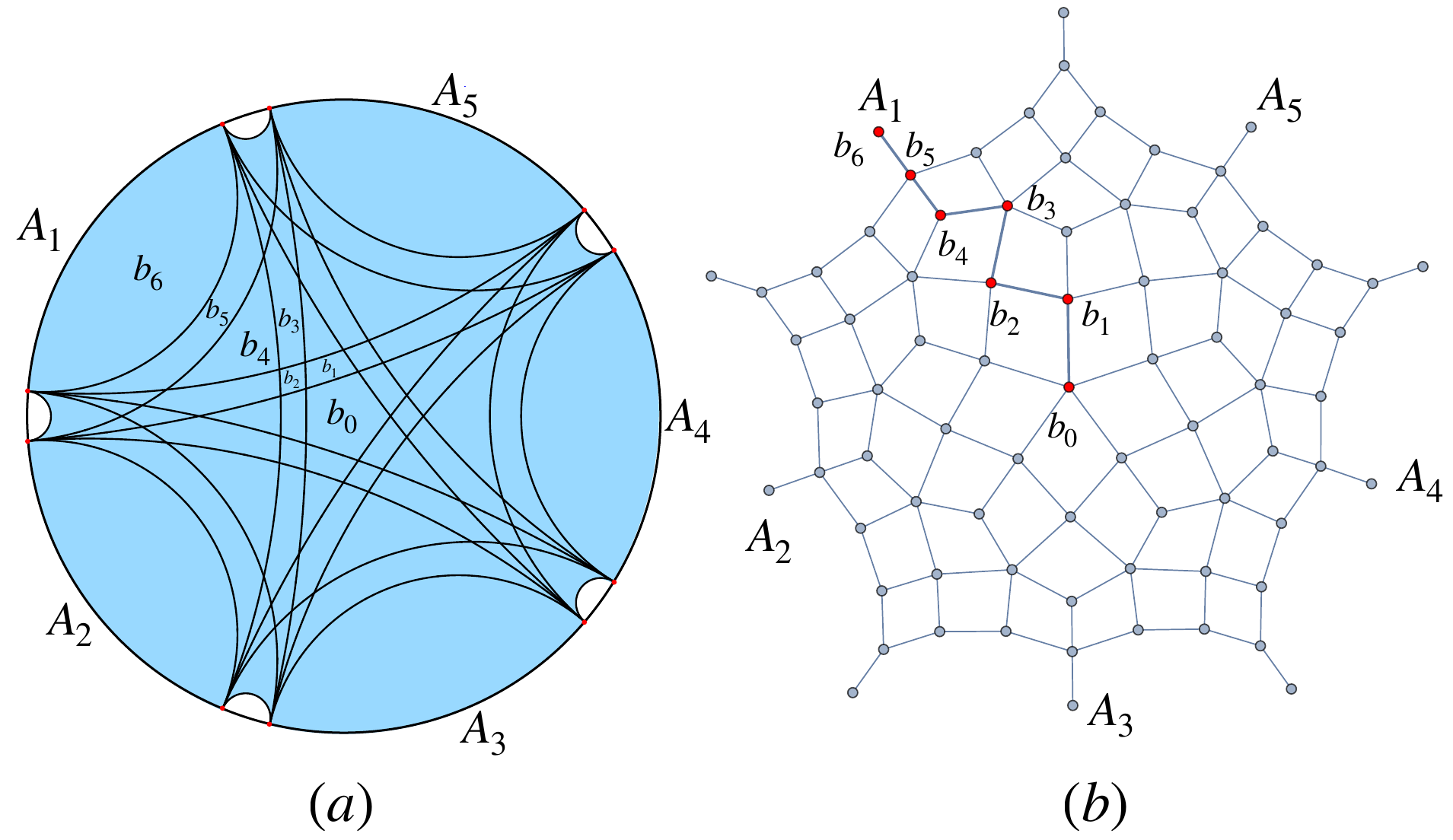}
    \caption{\small{(a) The constant time slice partitioned by RT surfaces of boundary subregions in $\mP[n=5]$ in the phase $\phi_4\leq \phi \leq 2\pi/5$. (b) The region graph for $n=5$. The vertices $b_0$, $b_1$, $b_2$, $b_3$, $b_4$, $b_5$ and $b_6$ in the path from the center of the region graph to the boundary vertex $A_1$ are colored red.}}
    \label{fig:examples_mixed_n-5_11_graphs}
\end{figure}

\begin{figure}[ht]
    \centering
    \includegraphics[width=0.9\linewidth]{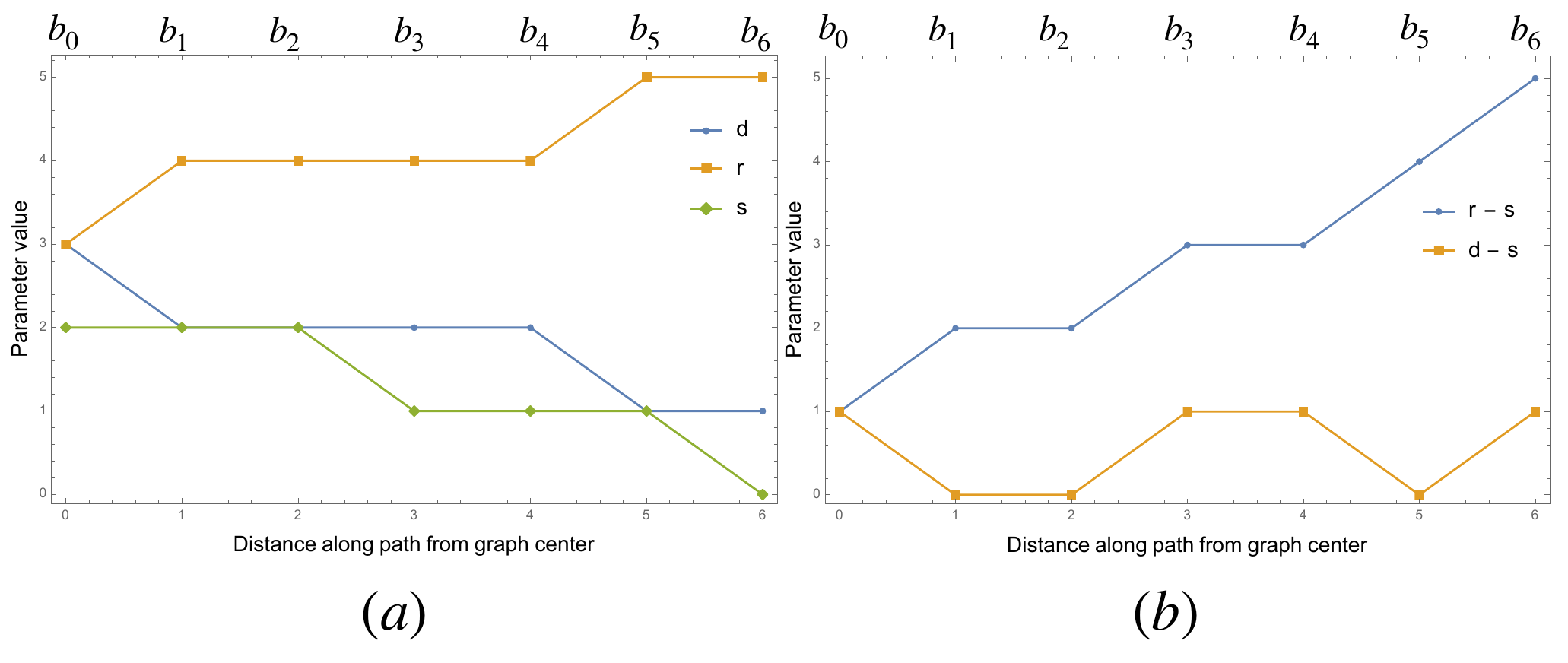}
    \caption{\small{$n=5$, mixed-state CHQSS. The plots of the parameters $d_b$, $r_b$, $s_b$, and the differences $r_b-s_b$ and $d_b-s_b$ as a function of the distance along the path from the center of the graph: (a) $d_b$, $r_b$, and $s_b$. (b) $r_b-s_b$, and $d_b-s_b$.}}
    \label{fig:examples_mixed_n-5_11_plots}
\end{figure}

\clearpage
\bibliographystyle{JHEP}
\bibliography{main}

\end{document}